\documentclass[draftclsnofoot,onecolumn]{IEEEtran}

\usepackage{amsthm,amsmath,amssymb,cite,citesort}
\usepackage{eucal}
\usepackage{xifthen}
\usepackage{mathtools}
\usepackage{enumerate}
\usepackage{microtype}
\usepackage{xspace}
\usepackage{bm}
\usepackage{fancyhdr}
\usepackage{lastpage}
\usepackage[center]{caption}
\usepackage{xcolor}
\allowdisplaybreaks

\newcommand{\mbs}[1]{\bm{#1}}
\newcommand{\vect}[1]{{\lowercase{\mbs{#1}}}}
\newcommand{\mat}[1]{{\uppercase{\mbs{#1}}}}

\renewcommand{\Bmatrix}[1]{\begin{bmatrix}#1\end{bmatrix}}
\newcommand{\Pmatrix}[1]{\begin{array}{ll}#1\end{array}}

\newcommand{\T}{{\scriptscriptstyle\mathsf{T}}}
\renewcommand{\H}{{\scriptscriptstyle\mathsf{H}}}

\newcommand{\cond}{\,\vert\,}
\renewcommand{\Re}[1][]{\ifthenelse{\isempty{#1}}{\operatorname{Re}}{\operatorname{Re}\left(#1\right)}}
\renewcommand{\Im}[1][]{\ifthenelse{\isempty{#1}}{\operatorname{Im}}{\operatorname{Im}\left(#1\right)}}

\newcommand{\hv}{\vect{h}}

\newcommand{\uv}{\vect{u}}
\newcommand{\vv}{\vect{v}}

\newcommand{\xv}{\vect{x}}
\newcommand{\yv}{\vect{y}}
\newcommand{\zv}{\vect{z}}

\newcommand{\etav}{\vect{\eta}}

\newcommand{\Lambdam}{\pmb{\Lambda}}
\newcommand{\Phim}{\pmb{\Phi}}

\newcommand{\Am}{\mat{a}}
\newcommand{\Bm}{\mat{b}}

\newcommand{\Gm}{\mat{g}}
\newcommand{\Hm}{\mat{h}}

\newcommand{\Km}{\mat{k}}

\newcommand{\Qm}{\mat{q}}

\newcommand{\Sm}{\mat{s}}
\newcommand{\Tm}{\mat{t}}
\newcommand{\Um}{\mat{u}}
\newcommand{\Vm}{\mat{V}}

\newcommand{\Xm}{\mat{x}}
\newcommand{\Ym}{\mat{y}}

\newcommand{\Cc}{{\mathcal C}}
\newcommand{\Dc}{{\mathcal D}}

\newcommand{\Hc}{{\mathcal H}}
\newcommand{\Ic}{{\mathcal I}}

\newcommand{\Uc}{{\mathcal U}}
\newcommand{\Wc}{{\mathcal W}}

\newcommand{\CC}{\mathbb{C}}

\newcommand{\Id}{\mat{\mathrm{I}}}

\newcommand{\CN}[1][]{\ifthenelse{\isempty{#1}}{\mathcal{N}_{\mathbb{C}}}{\mathcal{N}_{\mathbb{C}}\left(#1\right)}}

\renewcommand{\P}[1][]{\ifthenelse{\isempty{#1}}{\mathbb{P}}{\mathbb{P}\left(#1\right)}}
\newcommand{\E}[1][]{\ifthenelse{\isempty{#1}}{\mathbb{E}}{\mathbb{E}\left(#1\right)}}
\renewcommand{\det}[1][]{\ifthenelse{\isempty{#1}}{\mathrm{det}}{\mathrm{det}\left(#1\right)}}
\newcommand{\trace}[1][]{\ifthenelse{\isempty{#1}}{\mathrm{tr}}{\mathrm{tr}\left(#1\right)}}
\newcommand{\rank}[1][]{\ifthenelse{\isempty{#1}}{\mathrm{rank}}{\mathrm{rank}\left(#1\right)}}
\newcommand{\diag}[1][]{\ifthenelse{\isempty{#1}}{\mathrm{diag}}{\mathrm{diag}\left(#1\right)}}

\DeclarePairedDelimiter\Abs{\lvert}{\rvert^2}
\DeclarePairedDelimiter\norm{\lVert}{\rVert}
\DeclarePairedDelimiter\Norm{\lVert}{\rVert^2}
\DeclarePairedDelimiter\normf{\lVert}{\rVert_{\text{F}}}
\DeclarePairedDelimiter\Normf{\lVert}{\rVert^2_{\text{F}}}

\newcommand{\defeq}{\triangleq}
\newcommand{\eqdef}{\triangleq}

\newtheorem{remark}{Remark}[section]
\newtheorem{definition}{Definition}
\newtheorem{theorem}{Theorem}
\newtheorem{lemma}{Lemma}
\newtheorem{assumption}{Assumption}

\newcommand{\Es}{O_{\hat{S}}(1)}
\newcommand{\Eh}{O_{\hat{H}}(1)}
\newcommand{\Ex}{O_{X}}
\newcommand{\fdA}{f_{d\text{-}A}}
\newcommand{\fAd}{f_{A\text{-}d}}
\newcommand{\OdBC}{\mathcal{O}_d^{\mathrm{BC}}}
\newcommand{\IABC}{\mathcal{I}_A^{\mathrm{BC}}}
\newcommand{\IdBC}{\mathcal{I}_d^{\mathrm{BC}}}
\newcommand{\OdIC}{\mathcal{O}_d^{\mathrm{IC}}}
\newcommand{\IAIC}{\mathcal{I}_A^{\mathrm{IC}}}
\newcommand{\IdIC}{\mathcal{I}_d^{\mathrm{IC}}}
\newcommand{\AsetBC}{\mathcal{A_{\mathrm{BC}}}}
\newcommand{\AsetIC}{\mathcal{A_{\mathrm{IC}}}}
\newcommand{\DsetBC}{\mathcal{D_{\mathrm{BC}}}}
\newcommand{\DsetIC}{\mathcal{D_{\mathrm{IC}}}}
\newcommand{\Avec}{\Am}
\newcommand{\nk}{\min\{M_k,N_k\}}
\newcommand{\nj}{\min\{M_j,N_j\}}

\DeclareMathAlphabet{\mathcal}{OMS}{cmsy}{m}{n}

\usepackage{graphicx}
\usepackage{caption}
\usepackage{subcaption}
\usepackage{multirow}

\begin{document}
\title{The Degrees of Freedom Region of Temporally-Correlated MIMO Networks\\ with Delayed CSIT}
\author{
\authorblockN{Xinping Yi, \emph{Student Member, IEEE},  Sheng Yang, \emph{Member, IEEE}, \\ David Gesbert, \emph{Fellow, IEEE},   Mari Kobayashi, \emph{Member, IEEE}}
\thanks{X.~Yi and D. Gesbert are with the Mobile Communications Dept., EURECOM, 06560 Sophia Antipolis, France (email: \{xinping.yi, david.gesbert\}@eurecom.fr).}
\thanks{S.~Yang and M. Kobayashi are with the Telecommunications Dept., SUPELEC, 91190 Gif-sur-Yvette, France (e-mail: \{sheng.yang, mari.kobayashi\}@supelec.fr).}
\thanks{This work has been performed in the framework of the European research projects HARP and SHARING(FP7 ICT Objective 1.1) and HIATUS (FET-Open grant number: 265578), as well as the French ANR project FIREFLIES (ANR-10-INTB-0302).}
\thanks{A part of preliminary results of this work will be presented at ISIT 2013 \cite{Yi:2013ISIT}.}
}

\maketitle

\begin{abstract}
We consider the temporally-correlated Multiple-Input Multiple-Output (MIMO) broadcast channels (BC) and interference channels (IC) where the transmitter(s) has/have (i) {\em delayed} channel state information (CSI) obtained from a latency-prone feedback channel as well as (ii) imperfect {\em current} CSIT, obtained, e.g., from prediction on the basis of these past channel samples based on the temporal correlation. The degrees of freedom (DoF) regions for the two-user broadcast and interference MIMO networks with general antenna configuration under such conditions are fully characterized, as a function of the prediction quality indicator. Specifically, a simple unified framework is proposed, allowing to attain optimal DoF region for the general antenna configurations and current CSIT qualities. Such a framework builds upon {\em block-Markov encoding} with {\em interference quantization}, optimally combining the use of both outdated and instantaneous CSIT. A striking feature of our work is that, by varying the power allocation, \emph{every} point in the DoF region can be achieved with one single scheme. As a result, instead of checking the achievability of every corner point of the outer bound region, as typically done in the literature, we propose a new systematic way to prove the achievability. 
\end{abstract}

\section{Introduction}
While the capacity region of the Multiple-Input Multiple-Output (MIMO) broadcast channel (BC) was established in~\cite{MIMOBC2006}, the characterization of the capacity of Gaussian interference channel (IC) has been a long-standing open problem, even for the two-user single-antenna case. Recent progress sheds light on this problem from various perspectives, among which the authors in \cite{Jafar:2007IC} characterized the degrees of freedom (DoF) region, specializing to the large SNR regime, for the two-user MIMO IC. Such advance nevertheless suggests achievable schemes which require the full knowledge of channel state information (CSI) at both the transmitter and receiver sides. In practice, however, the acquisition of perfect CSI at the transmitters is difficult, if not impossible, especially for fast fading channels. The CSIT obtained via feedback suffers from delays, which renders the available CSIT feedback possibly fully obsolete (i.e., uncorrelated with the current true channel) under the fast fading channel and, seemingly non-exploitable in view of designing the spatial precoding.

Recently, this common accepted viewpoint in such scenario (referred to
as ``delayed CSIT'') was challenged by an interesting information
theoretic work \cite{MAT2012TIT}, in which a novel scheme (termed here as
``MAT alignment'') was proposed for the MISO BC to demonstrate that even
the completely outdated channel feedback is still useful. 
The precoders are designed achieving strictly better DoF than what is obtained without
any CSIT. The essential ingredient for the proposed scheme in
\cite{MAT2012TIT} lies in the use of a multi-slot protocol initiating
with the transmission of unprecoded information symbols to the user
terminals, followed by the \emph{analog} forwarding of the interference
created in the previous time slots. Most recently, generalizations under
the similar principle to the MIMO BC~\cite{Vaze:2011BC}, MIMO IC
\cite{Vaze:2012IC,Ghasemi:2011} settings, among others, were also
addressed, where the DoF regions are fully characterized with arbitrary
antenna configurations, again establishing DoF strictly beyond the ones obtained without CSIT~\cite{Huang:2012,Guo:2012noCSIT,Vaze:2012noCSIT} but below the ones with perfect CSIT \cite{MIMOBC2006,Jafar:2007IC}. Note that other recent interesting lines of work combining instantaneous and delayed forms of feedback were reported in \cite{Tandon:2012alter,Lee:2012}.

Albeit inspiring and fascinating from a conceptual point of view, these
works made an assumption that the channel is independent and identically
distributed (i.i.d.) across time, where the delayed CSIT bears no
correlation with the current channel realization. Hence, these results
pessimistically consider that no estimate for the  {\em current} channel
realization is producible to the transmitter. Owing to the finite
Doppler spread behavior of fading channels, it is however the case in
many real life situations that the past channel realizations can provide
\emph{some} information about the current one. Therefore a scenario
where the transmitter is endowed with delayed CSI in addition to some
(albeit imperfect) estimate of the current channel is of practical
relevance. Together with the delayed CSIT, the benefit of such imperfect
current CSIT was first exploited in \cite{Kobayashi:2012} for the MISO
BC whereby a novel transmission scheme was proposed which improves over
pure MAT alignment in  constructing precoders based on delayed
\emph{and} current CSIT estimate. The full characterization of the
optimal DoF for this hybrid CSIT was later reported in
\cite{Yang:2013,Gou:2012} for the MISO BC under this setting. The key idea
behind the schemes (termed hereafter as ``$\alpha$-MAT alignment'') in \cite{Kobayashi:2012,Yang:2013,Gou:2012} lies in
the modification of the MAT alignment such that i) the initial time slot
involves transmission of {\em precoded} symbols, which enables to reduce the power of mutual interferences and efficiently compress them; ii) the subsequent slots perform a digital transmission of quantized residual interferences together with new private symbols. Most recently, this philosophy was extended to the
MIMO networks (BC/IC) but only with symmetric antenna
configuration~\cite{Yi:2012}, as well as the $K$-user MISO case
\cite{Kerret:2013}. The generalization to the MISO BC with different
qualities of imperfect current CSIT was also studied in \cite{Chen:2012}. 
Remarkably, the authors of \cite{Chen:2012} showed that, in order to
balance the asymmetry of the CSIT quality, an infinite number of time slots
are required. As such, they extended the number of phases of the
$\alpha$-MAT alignment~\cite{Yang:2013} to infinity and varied the length of each phase, such that the
optimal DoF can be achieved. 

Unfortunately, extending the previous results to the MIMO case with arbitrary
antenna configurations is not a trivial step, even with the symmetric
current CSIT quality assumption. 
The main challenges are two-fold: (a) the extra spatial dimension at
the receiver side introduces a non-trivial tradeoff between useful
signal and mutual interference,  
and (b) the asymmetry of receive antenna
configurations results in the discrepancy of
common-message-decoding capability at different receivers. In
particular, the total number of streams that can be delivered as common
messages to both receivers is inevitably limited by the weaker one~(i.e., 
with fewer antennas). Such a
constraint prevents the system from achieving the optimal DoF of the
symmetric case by simply extending the previous schemes developed
in~\cite{Yi:2012}.

To counter these new challenges posed by the asymmetry of antenna
configurations, we develop a
new strategy that balances the discrepancy of common-message-decoding
capability at two receivers.  This allows us to fully characterize the DoF
region of both MIMO BC and MIMO IC, achieved by a unified and simple
scheme built upon {\em block-Markov encoding}. This
encoding concept was first introduced in \cite{block-Markov} for relay
channels and then became a standard tool for communication problems
involving interaction between nodes, such as feedback~(e.g.,
\cite{Ozarow,Suh:2011}) or user cooperation~(e.g., \cite{Willems:1985}).
It turns out that our problem with both delayed and instantaneous CSIT,
closely related to \cite{Ozarow}, can also be solved with this scheme.
As it will become clear later, in each block, the transmitter
superimposes the common information about the interference created in the past
block~(due to the imperfect instantaneous CSIT) on the new private
information~(thus creating new interference). At the receiver side,
backward decoding is employed, i.e., the decoding of each block
relies on the common side information from decoding of the future block.  
Due to the repetitive nature in each block, 
the proposed scheme can be uniquely characterized
with the parameters such as the power allocation and rate splitting of 
the superposition. 
Surprisingly enough, our block-Markov scheme can
also include the asymmetry of current CSIT with a simple parameter
change, and thus somehow balance the global asymmetry, i.e., antenna
asymmetry and CSIT asymmetry, in the system. 

Overall, our results allow to bridge between previously reported CSIT scenarios such as the pure delayed CSIT of \cite{MAT2012TIT,Vaze:2011BC} and the pure instantaneous CSIT scenarios \cite{MIMOBC2006,Jafar:2007IC} for the MIMO setting.  We tackle both the BC and IC configurations as we point out the tight connection between the DoF achieving transmission strategies in both settings. More specifically, we obtain the following key results:
\begin{itemize}
  \item We establish outer bounds on the DoF region for the two-user
    temporally-correlated MIMO BC and IC with perfect delayed and
    imperfect current CSIT, as a function of the current CSIT quality
    \emph{exponent}. {By introducing {a virtual received signal} for the IC,
    we nicely link the outer bound to that of the BC, arriving at a
    similar outer bound result for both cases. In addition to the
    genie-aided bounding techniques and the application of the extremal
    inequality in \cite{Yang:2013}, we develop a set of upper and lower
    bounds of ergodic capacity for MIMO channels, which is essential for
    the MIMO case but not extendible from MISO.}
  \item We propose a unified framework relying on block-Markov encoding
    uniquely parameterized by the rate splitting and power allocation, 
    by which the optimal
    DoF region confined by the outer bound are achievable with perfect
    delayed plus imperfect current CSIT. 
    For any antenna and current CSIT settings, every point in the outer
    bound region can be achieved with one single scheme.    
    For instance, the MIMO BC with $M=3$, $N_1=2$ and $N_2=1$ achieves
    optimal sum DoF $\frac{15+4\alpha_1+2\alpha_2}{7}$ when
    $3\alpha_1-2\alpha_2 \le 1$ and $\frac{7+2\alpha_2}{3}$ otherwise,
    where $\alpha_1$ and $\alpha_2$ are imperfect current CSIT qualities
    for both users' channels. This smoothly connects three special
    cases: the case with pure delayed CSIT \cite{Vaze:2011BC}
    ($\alpha_1=\alpha_2=0$), that with perfect current CSIT
    \cite{MIMOBC2006} ($\alpha_1=\alpha_2=1$), and that with perfect
    CSIT at Receiver~1 and delayed CSIT at Receiver~2 \cite{Tandon:2012ISWCS}
    $(\alpha_1=1, \alpha_2=0)$. 
  \item We propose a new systematic way to prove the achievability. In
    the proposed framework, the achievability region is defined by the
    decodability conditions in terms of the rate splitting and power
    allocation. The achievability is proved by mapping the
    outer bound region into a set of proper rate and power allocation and
    showing that this set lies within the decodability region. This contrasts
    with most existing proofs in the literature where the achievability
    of each corner point is checked. 
\end{itemize}
It is worth noting that our results embrace the previously reported
particular cases: the perfect CSIT setting
\cite{MIMOBC2006,Jafar:2007IC} (i.e., current CSIT of perfect quality),
the pure delayed CSIT setting \cite{Vaze:2012IC,Ghasemi:2011}~(i.e.,
current CSIT of zero quality), the partial/hybrid CSIT MIMO BC/IC case
\cite{RIA,Tandon:2012ISWCS,Vaze:2012Hybrid}~(with perfect CSIT at one receiver
and delayed CSIT at the other one), and the special MISO case
\cite{Kobayashi:2012,Yang:2013,Gou:2012} with $N_1=N_2=1$, symmetric
MIMO case \cite{Yi:2012}, as well as the MISO case with asymmetric
current CSIT qualities~\cite{Chen:2012}. 

The rest of the paper is organized as follows. We present the system
model and assumptions in the coming section, followed by the main
results on DoF region characterization for both MIMO BC and MIMO IC
cases in Section III. Some illustrative examples of the achievability schemes are provided in Section IV, followed  by the general formulation in Section V. In Section VI, we present the proofs of outer bounds. Finally, we conclude the paper in Section VII.

\textbf{Notation}: Matrices and vectors are represented as uppercase and
lowercase letters, respectively. Matrix transport, Hermitian transport, inverse,
rank, determinant and the Frobenius norm of a matrix are denoted by
$\Am^\T$, $\Am^\H$, $\Am^{-1}$, $\rank(\Am)$, $\det(\Am)$ and
$\normf{\Am}$, respectively. $\Am_{[k_1:k_2]}$ represents the submatrix
of $\Am$ from $k_1$-th row to $k_2$-th row when $k_1 \le k_2$.
$\hv^{\bot}$ is the normalized orthogonal component of any nonzero
vector $\hv$. We use $\Id_M$ to denote an $M\times M$ identity matrix
where the dimension is omitted whenever confusion is not probable.  The
approximation $f(P) \sim g(P)$ is in the sense of $\lim_{P \to \infty}
\frac{f(P)}{g(P)}=C$, where $C>0$ is a constant that does not scale as
$P$. Partial ordering of Hermitian matrices is denoted by $\succeq$ and
$\preceq$, i.e., $\Am \preceq \Bm$ means $\Bm-\Am$ is positive
semidefinite. Logarithms are in base 2. $(x)^+$ means $\max\{x,0\}$, and
$\mathbb{R}_{+}^n$ represents the set of $n$-tuples of non-negative real
numbers. $f = O(g)$ follows the standard Landau notation, i.e.,
$\lim \frac{f}{g} \le C$ where the limit depends on the context. With
some abuse of notation, we use $\Ex(g)$ to denote any $f$ such that
$\E_X (f) = O(\E_X(g))$. 
 Finally, the range or null spaces mentioned in this paper refer to the column spaces.

\section{System Model}
\subsection{Two-user MIMO Broadcast Channel}
For a two-user $(M,N_1,N_2)$ MIMO broadcast channel (BC) with $M$ antennas at the Transmitter and $N_i$ antennas at the Receiver~$i$, the discrete time signal model is given by
\begin{align}
  \yv_i(t) &=\Hm_i(t) \xv(t) + \zv_i(t)
\end{align}
for any time instant $t$, where $\Hm_i(t) \in \CC^{N_i\times M}$ is the
channel matrix for Receiver~$i$ ($i=1,2$); $\zv_i(t) \sim \CN[0,\Id_{N_i}]$ is
the normalized additive white Gaussian noise (AWGN) vector at the Receiver~$i$
and is independent of channel matrices and transmitted signals; the coded input signal $\xv(t) \in \CC^{M\times 1}$ is subject to the power
constraint $\E\bigl( \norm{\xv(t)}^2 \bigr)  \le P$, $\forall\,t$. 

\subsection{Two-user MIMO Interference Channel}

For a two-user $(M_1,M_2,N_1,N_2)$ MIMO interference channel (IC) with $M_i$ antennas at Transmitter~$i$ and $N_j$ antennas at Receiver~$j$, for $i,j=1,2$, the discrete time signal model is given by
\begin{align}
\yv_i(t) &= \Hm_{i1}(t) \xv_1(t) + \Hm_{i2}(t) \xv_2(t) + \zv_i(t)
\end{align}
for any time instant $t$, where $\Hm_{ji}(t) \in \CC^{N_j\times M_i}$
($i,j=1,2$) is the channel matrix between Transmitter~$i$ and Receiver~$j$; the coded
input signal $\xv_i(t) \in \CC^{M_i\times 1}$ is subject to the power
constraint $\E  \bigl( \norm{\xv_i(t)}^2 \bigr)  \le P$ for $i=1,2$, $\forall\,t$.  

In the rest of this paper, we refer to MIMO BC/IC as MIMO networks. For
notational brevity, we define  the ensemble of channel matrices, i.e.,
$\Hc(t) \defeq \{\Hm_1(t),\Hm_2(t)\}$ (resp.~$\Hc(t) \defeq
\{\Hm_{11}(t),\Hm_{21}(t),\Hm_{12}(t),\Hm_{22}(t)\}$), as the channel state for BC (resp.~IC). We further define $\Hc^k \defeq \{\Hc(t)\}_{t=1}^k$, and $\hat{\Hc}^k \defeq \{\hat{\Hc}(t)\}_{t=1}^k$, where $k=1,\cdots,n$. 

\subsection{Assumptions and Definitions}

\begin{assumption} [perfect delayed and imperfect current CSIT]
  \label{ass:CSI}
  At each time instant $t$, the transmitters know perfectly the delayed CSI $\Hc^{t-1}$, and obtains an imperfect estimate of the current CSI $\hat{\Hc}(t)$, which could, for instance, be produced by standard prediction based on past samples. The current CSIT estimate is modeled by
  \begin{align}
    \Hm_i(t) &= \hat{\Hm}_i(t) + \tilde{\Hm}_i(t) \\
    \Hm_{ij}(t) &= \hat{\Hm}_{ij}(t) + \tilde{\Hm}_{ij}(t)
  \end{align}
  for BC and IC, respectively, where estimation error $\tilde{\Hm}_i(t)$
  (resp.~$\tilde{\Hm}_{ij}(t)$) and the estimate
  $\hat{\Hm}_i(t)$~(resp.~$\hat{\Hm}_{ij}(t)$) are mutually independent,
  and each entry is assumed\footnote{
  We make the above assumption on the fading distribution to simplify the 
  the presentation, although the results can be applied to a broader class
  of distributions. } to be $\CN[0, \sigma_i^2]$ and
  $\CN[0, 1-\sigma_i^2]$. Further, we assume the following Markov chain
  \begin{align} \label{eq:markov-chain}
    ({\Hc}^{t-1},\hat{\Hc}^{t-1}) \to \hat{\Hc}(t) \to {\Hc}(t), \quad \forall t,
  \end{align}
  which means ${\Hc}(t)$ is independent of
  $({\Hc}^{t-1},\hat{\Hc}^{t-1})$ conditional on $\hat{\Hc}(t)$.
  Furthermore, at the end of the transmission, i.e., at time instant $n$, the receivers know perfectly ${\Hc}^{n}$ and $\hat{\Hc}^{n}$.
\end{assumption}
It readily follows that, for any fat submatrix $\Hm$ of $\Hm_i$ or
$\Hm_{ij}$, $\E(\log\det(\Hm \Hm^\H))>-\infty$ and
$\E(\log\det(\hat{\Hm} \hat{\Hm}^\H)) = O(1)$ when $\sigma_i^2$ goes
$0$.

The assumption on the CSI at the receiver (CSIR) is in accordance with previous works with delayed CSIT, and 
does not add any limitation over the assumption made in \cite{MAT2012TIT,Vaze:2011BC,Vaze:2012IC}. We point out that only local CSIT/CSIR (the channel links with which the node is connected) is really helpful {and leads to the same result}. Nevertheless, we assume the CSIT/CSIR to be available in a global fashion for presentation simplicity.

We are interested in characterizing the degrees of freedom (DoF) of the above system as functions of the quality of current CSIT, thus bridging between the two previously investigated extremes which are the perfect instantaneous CSIT and the fully outdated (no instantaneous) CSIT cases. As it was established in previous works~\cite{Kobayashi:2012,Yang:2013}, the imperfect current CSIT has beneficial value (in terms of improving the DoF) only if the CSIT estimation error decays at least exponentially with the SNR or faster. Thus it is reasonable to study the regime by which the CSIT quality can be parameterized by an indicator $\alpha_i \ge 0$ such that:
\begin{align}
  \alpha_i \defeq  -\lim_{P \to \infty} \frac {\log \sigma_i^2}{\log P} \label{eq:alpha-def}
\end{align}
if the limit exists.
This $\alpha_i$ indicates the quality of current CSIT at high SNR. While
$\alpha_i=0$ reflects the case with no current CSIT, $\alpha_i \to
\infty$ corresponds to that with perfect instantaneous CSIT. As a matter
of fact, when $\alpha_i \ge 1$, the quality of the imperfect current
CSIT is sufficient to avoid the DoF loss, and ZF precoding with this
imperfect CSIT is able to achieve the maximum DoF~\cite{Caire:2010}.
Therefore, we focus on the case $\alpha_i \in [0,1]$ henceforth. The
connections between the above model and the linear prediction over
existing time-correlated channel models with prescribed user mobility
are highlighted in \cite{Kobayashi:2012,Yang:2013}. According to the
definition of the estimated current CSIT, we have $\E
\bigl(\Abs{{\hv}_k^{\H}(t) \hat{\hv}_k^{\perp}(t)} \bigr) = \sigma_i^2 \sim P^{-\alpha_i}$, with ${\hv}_k^\H$ representing any row of channel matrices $\Hm_i(t)$ (resp.~$\Hm_{ij}(t)$), and $\hat{\hv}_k^\H$ being its corresponding estimate.

A rate pair $(R_1,R_2)$ is said to be {\em achievable} for the two-user MIMO networks with perfect delayed and imperfect current CSIT if there exists a $\left(2^{nR_1},2^{nR_2},n\right)$ code scheme consists of:
\begin{itemize}
  \item two message sets $\Wc_1 \defeq [1:2^{nR_1}]$ and $\Wc_2 \defeq [1:2^{nR_2}]$, from which two independent messages $W_1$ and $W_2$ intended respectively to the Receiver~1 and Receiver~2 are uniformly chosen;
  \item one encoding function for (each) transmitter: 
  \begin{align} \label{eq:enc-fun}
    \Pmatrix{
    \text{BC:} \quad \xv(t) = f_t
    \bigl(W_1,W_2,{\Hc}^{t-1},\hat{\Hc}^{t}\bigr)
    \\
    \text{IC:} \quad \xv_i(t) = f_{i,t}
    \bigl(W_i,{\Hc}^{t-1},\hat{\Hc}^{t}\bigr), ~ i=1,2;}
    \end{align}
  \item one decoding function at the corresponding receiver,
    \begin{align}
    \hat{W}_j &= g_{j} \bigl( \Ym_j^n,\Hc^{n},\hat{\Hc}^{n} \bigr),~j=1,2 \label{eq:dec-fun}
    \end{align}
    for the Receiver~$j$, where $\Ym_j^n \defeq \{\yv_j(t)\}_{t=1}^{n}$,
\end{itemize}
such that the average decoding error probability $P_{e}^{(n)}$, defined as
$
  P_{e}^{(n)} \defeq \P\bigl((W_1, W_2) \neq (\hat{W}_1,\hat{W}_2)
  \bigr) ,
$
vanishes as the code length $n$ tends to infinity. The capacity region $\Cc$ is defined as the set of all achievable rate pairs. Accordingly, the DoF region can be defined as follows:
\begin{definition} [the degrees of freedom region]
  The degrees of freedom (DoF) region for the two-user MIMO network is defined as
  \begin{align}
    \Dc &= \left\{ (d_1,d_2)\in \mathbb{R}_{+}^2 |~ \forall (w_1,w_2) \in \mathbb{R}_{+}^2, w_1d_1+w_2d_2 \le \limsup_{P \to \infty} \left( \sup_{(R_1,R_2) \in \Cc} \frac{w_1R_1+w_2R_2}{\log P}\right) \right\}.
  \end{align}
\end{definition}

\section{Main Results}
According to the assumptions and definitions in the previous section, the main results of this paper are stated as the following two theorems:
\begin{theorem}\label{theorem:BC}
  For the two-user $(M,N_1,N_2)$ MIMO BC with delayed and imperfect current CSIT, the optimal DoF region $\{(d_1,d_2)| (d_1,d_2) \in \mathbb{R}_{+}^2\}$ is characterized by
  \begin{subequations} \label{eq:outer-bound}
    \begin{align}
      d_1 &\le \min\{M,N_1\},\label{eq:bc-outer-bound-1}\\
      d_2 &\le \min\{M,N_2\},\label{eq:bc-outer-bound-2}\\
      d_1 + d_2 &\le \min\{M,N_1+N_2\},\label{eq:bc-outer-bound-3}\\
      \frac{d_1}{\min\{M,N_1\}} + \frac{d_2}{\min\{M,N_1+N_2\}} &\le 1+\frac{\min\{M,N_1+N_2\}-\min\{M,N_1\}}{\min\{M,N_1+N_2\}}\alpha_1, \label{eq:bc-outer-bound-4}\\
      \frac{d_1}{\min\{M,N_1+N_2\}} + \frac{d_2}{\min\{M,N_2\}} &\le 1+\frac{\min\{M,N_1+N_2\}-\min\{M,N_2\}}{\min\{M,N_1+N_2\}}\alpha_2, \label{eq:bc-outer-bound-5}
    \end{align}
  \end{subequations}
  where $\alpha_i \in [0,1]$ $(i=1,2)$ indicates the current CSIT quality exponent of Receiver~$i$'s channel.
\end{theorem}
\begin{proof}
The proof of achievability will be presented in Section IV showing some insight with toy examples, and in Section V for the general formulation. The converse proof will be given in Section VI focusing on \eqref{eq:bc-outer-bound-4} and \eqref{eq:bc-outer-bound-5}, because the first three bounds correspond to the upper bounds under perfect CSIT and thus hold trivially under delayed and imperfect current CSIT.
\end{proof}

\begin{remark}
This result yields a number of previous results as special cases:~the delayed CSIT case \cite{Vaze:2011BC} for $\alpha_1=\alpha_2=0$, where the sum DoF bound \eqref{eq:bc-outer-bound-3} is inactive; perfect CSIT case \cite{MIMOBC2006} for $\alpha_1=\alpha_2=1$, where the weighted sum DoF bounds \eqref{eq:bc-outer-bound-4} and \eqref{eq:bc-outer-bound-5} are inactive; partial CSIT (i.e., perfect CSIT for one channel and delayed CSIT for the other one) case \cite{Tandon:2012ISWCS} for $\alpha_1=1,\alpha_2=0$, where only \eqref{eq:bc-outer-bound-2} and \eqref{eq:bc-outer-bound-5} are active; delayed CSIT in MISO BC for $N_1=N_2=1$ \cite{Yang:2013,Gou:2012,Chen:2012}.
\end{remark}

Before presenting the optimal DoF region for MIMO IC, we specify two conditions.
\begin{definition} [Condition $C_k$]
Given $k \in \{1,2\}$,  condition $C_k$ holds, indicating the following inequality
\begin{align}
M_k \ge N_j, \quad M_j < N_k, \quad M_1+M_2 > N_1 +N_2
\end{align}
is true, $\forall~j \in \{1,2\},~j \ne k$. 
\end{definition}

\begin{remark} 
This definition that points out the existence of the corresponding outer bound, is different from that in \cite{Vaze:2012IC}, in which the condition implies the activation of the outer bounds.
\end{remark}

\begin{theorem}\label{theorem:IC}
  For the two-user $(M_1,M_2,N_1,N_2)$ MIMO IC with delayed and imperfect current CSIT, the optimal DoF region $\{(d_1,d_2)| (d_1,d_2) \in \mathbb{R}_{+}^2\}$ is characterized by
  \begin{subequations} \label{eq:outer-bound-ic}
    \begin{align}
      d_1 &\le \min\{M_1,N_1\},\\
      d_2 &\le \min\{M_2,N_2\}, \label{eq:ic-outer-bound-2} \\
      d_1 + d_2 &\le \min\{M_1+M_2,N_1+N_2,\max\{M_1,N_2\},\max\{M_2,N_1\}\}, \label{eq:ic-outer-bound-3} \\
      \frac{d_1}{\min\{M_2,N_1\}} + \frac{d_2}{\min\{M_2,N_1+N_2\}} &\le \frac{\min\{N_1,M_1+M_2\}}{\min\{M_2,N_1\}}+\frac{\min\{M_2,N_1+N_2\}-\min\{M_2,N_1\}}{\min\{M_2,N_1+N_2\}}\alpha_1, \label{eq:ic-outer-bound-4}\\
      \frac{d_1}{\min\{M_1,N_1+N_2\}} + \frac{d_2}{\min\{M_1,N_2\}} &\le \frac{\min\{N_2,M_1+M_2\}}{\min\{M_1,N_2\}}+\frac{\min\{M_1,N_1+N_2\}-\min\{M_1,N_2\}}{\min\{M_1,N_1+N_2\}}\alpha_2,  \label{eq:ic-outer-bound-5}\\
      d_1 + \frac{N_1+2N_2-M_2}{N_2}d_2 &\le N_1+N_2 + (N_1-M_2)\alpha_2, \quad \text{if $C_1$ holds}  \label{eq:ic-outer-bound-6} \\
      d_2 + \frac{N_2+2N_1-M_1}{N_1}d_1 &\le N_1+N_2 + (N_2-M_1)\alpha_1, \quad \text{if $C_2$ holds}   \label{eq:ic-outer-bound-7}
    \end{align}
  \end{subequations}
  where $\alpha_i \in [0,1]$ $(i=1,2)$ indicates the current CSIT quality exponent corresponds to Receiver~$i$.
\end{theorem}
\begin{proof}
The general formulation of achievability will be presented in Section V, and the converse will be given in Section VI. For the converse, the first three inequalities correspond to the outer bounds for the case of perfect CSIT, which should also hold for our setting. Hence, it is sufficient to prove the last four bounds. Due to the symmetry property of the bounds \eqref{eq:ic-outer-bound-4} and \eqref{eq:ic-outer-bound-5},   \eqref{eq:ic-outer-bound-6} and \eqref{eq:ic-outer-bound-7}, it is sufficient to prove the bounds \eqref{eq:ic-outer-bound-4} and \eqref{eq:ic-outer-bound-6}.
\end{proof}

\begin{remark}
 Some previous reported results can be regarded as the special cases of our results: delayed CSIT case \cite{Vaze:2012IC} for $\alpha_1=\alpha_2=0$; perfect CSIT case \cite{Jafar:2007IC} for $\alpha_1=\alpha_2=1$, where the weighted sum DoF bounds \eqref{eq:ic-outer-bound-4}-\eqref{eq:ic-outer-bound-7} are inactive; hybrid CSIT (i.e., perfect CSIT for one channel and delayed CSIT for the other one) case \cite{Vaze:2012Hybrid} for $\alpha_1=1,\alpha_2=0$, where  the bounds \eqref{eq:ic-outer-bound-5} and \eqref{eq:ic-outer-bound-6} are active.
\end{remark}

\section{Achievability: Toy Examples}

To introduce the main idea of our achievability scheme, we 
revisit  MAT \cite{MAT2012TIT} and $\alpha$-MAT alignment \cite{Kobayashi:2012,Yang:2013,Gou:2012} for the case of MISO BC in Section IV.A, followed by an alternative way built
on {\em block-Markov encoding and backward decoding} in Section IV.B, as well as some examples in Section IV.C and IV.D showing that block-Markov encoding allows us to balance the asymmetry both in current CSIT qualities and antenna configurations. 
Although MAT \cite{MAT2012TIT} and $\alpha$-MAT alignment \cite{Kobayashi:2012,Yang:2013,Gou:2012} appear to be conceptually different, these schemes boil down into a {\em single} block-Markov encoding scheme (of an infinite number of constant length blocks). In fact, both schemes can be represented exactly in the same manner with different parameters.

\subsection{MAT v.s. $\alpha$-MAT Revisit}
Let us take the simplest antenna configuration, i.e., $(2,1,1)$ BC, as
an example. Recall that both MAT and $\alpha$-MAT deliver symbol under
the same structure. Specifically, in the first phase (Phase I), two independent
messages $w_1$ and $w_2$ are encoded into two independent vectors
$\uv_1(w_1)$ and
$\uv_2(w_2)$ with different covariance matrices $\Qm_{1}\defeq\E[\uv_1 \uv_1^\H]$
and $\Qm_{2}\defeq\E[\uv_2 \uv_2^\H]$. The sum of these vectors are sent out, i.e.,
\begin{align}
\xv[1] = \uv_1 + \uv_2, \quad s.t.~\left\{\Pmatrix{\text{MAT:} & \quad
\Qm_{1} = \Qm_{2} = P \Id\\
\alpha\text{-MAT:} & \quad \Qm_{1} = P_1 \Phim_{\hat{h}_2} + P_2
\Phim_{\hat{h}_2^\bot}, ~ \Qm_{2} = P_1 \Phim_{\hat{h}_1} + P_2
\Phim_{\hat{h}_1^\bot}}\right.
\end{align}
where $P_1 \sim P^{1-\alpha}$, $P_2=P-P_1 \sim P$, $\forall\,\alpha>0$, and
$\Phim_{h} \defeq \frac{\hv \hv^\H}{\Norm{\hv}}$.  
Each receiver experiences some interference caused by the symbols
dedicated to the other receiver
\begin{align}
\left\{\Pmatrix{\eta_1 \defeq \hv_1^\H \uv_2, \\ \eta_2 \defeq \hv_2^\H
\uv_1}\right. \quad s.t.~\left\{\Pmatrix{\text{MAT:} &  \E[\Abs{\eta_i}]
\sim P \\ \alpha\text{-MAT:} &  \E[\Abs{\eta_i}] \sim P^{1-\alpha}} \right.
\end{align}
Then, the task of the second phase is to multicast the interferences
$(\eta_1, \eta_2)$ to \emph{both} receivers. The main difference between
the MAT and $\alpha$-MAT lies in the way in which the interferences are
sent. While the analog version of $\eta_k$ is sent in two slots with
MAT, the digitized version is sent with $\alpha$-MAT instead.  
Note that the covariance matrices $\Qm_{1}$ and $\Qm_{2}$, or
equivalently, the spatial precoding and power allocation, of
$\alpha$-MAT are such that the mutual interferences $(\eta_1, \eta_2)$
have a reduced power level $P^{1-\alpha}$. According to the rate-distortion
theorem~\cite{Cover_Thomas}, each interference $\eta_k$, $k=1,2$, can be compressed with 
a source codebook of size $P^{1-\alpha}$ or $(1-\alpha)\log P$ bits into
an index $l_k$, in such a way that the average distortion between $\eta_k$
and the source codeword $\hat{\eta}_k(l_k)$ is comparable to the AWGN
level~\cite{Yang:2013}. Then, the index $l_k$ is encoded
with a channel codebook into a codeword $\xv_c(l_k)\sim P \Id_2$ and sent as the 
common message to both receivers. Thanks to the reduced range of
$l_k$, there is still room to transmit private messages. The structure
of the two slots in the second phase (Phase II) is    
\begin{align}
\left\{\Pmatrix{\text{MAT:} & \xv[2] = \vv_k \eta_k,  
\\ \alpha\text{-MAT:}   & \xv[2] = \xv_c(l_k) +
\uv_{p1} +  \uv_{p2}} \right.
\end{align}
where $k=1,2$, $\vv_k$ is a randomly chosen vector; the covariances of
the private signals $\uv_{p1}$ and $\uv_{p2}$ are respectively
$\Qm_{u_{p1}} =
P^{\alpha} \Phim_{\hat{h}_2^\perp}$ and $\Qm_{u_{p2}} = P^{\alpha}
\Phim_{\hat{h}_1^\perp}$ in such a way that they are drown in the AWGN at
the unintended receivers without creating noticeable
interferences~(at high SNR). At receiver~$k$, the common messages
$l_1$ and $l_2$ are first decoded from the two slots in Phase~II, by
treating the private signal $\uv_{p1}$ or $\uv_{p2}$ as noise. The common
messages are then used to~1)~reconstruct $\eta_1$ and $\eta_2$ that will
be used with the received signal in Phase I to decode $w_k$ and recover
$2-\alpha$ DoF, and 2)~to
reconstruct $\xv_c(l_k)$ and remove it from the received signals in
Phase~II so as to decode the private messages and recover $2\alpha$
DoF~(in two slots). In the end, $2-\alpha+2\alpha = 2+\alpha$ of DoF per
user is achievable in three slots, yielding an average DoF of
$\frac{2+\alpha}{3}$ per user.  
The interested readers may refer to \cite{Yang:2013} for more details of $\alpha$-MAT alignment.

\subsection{An Alternative: Block-Markov Implementation}

In fact, both MAT and $\alpha$-MAT can be implemented in a block-Markov
fashion, the concept of which is shown in Fig.~\ref{fig:BME} for $\alpha=0$. 
The common message $\xv_c(l_{b-1})$ comes from the previous block $b-1$, and $\uv_k(w_{kb})$ is the new private message dedicated to Receiver~$k$ $(k=1,2)$.
Essentially, we
``squeeze'' the Phase~II of block $b-1$ and the Phase~I of
block $b$ into one single block, with proper power and rate scaling.

\begin{figure}[htb]
\centering
\includegraphics[width=0.8\columnwidth]{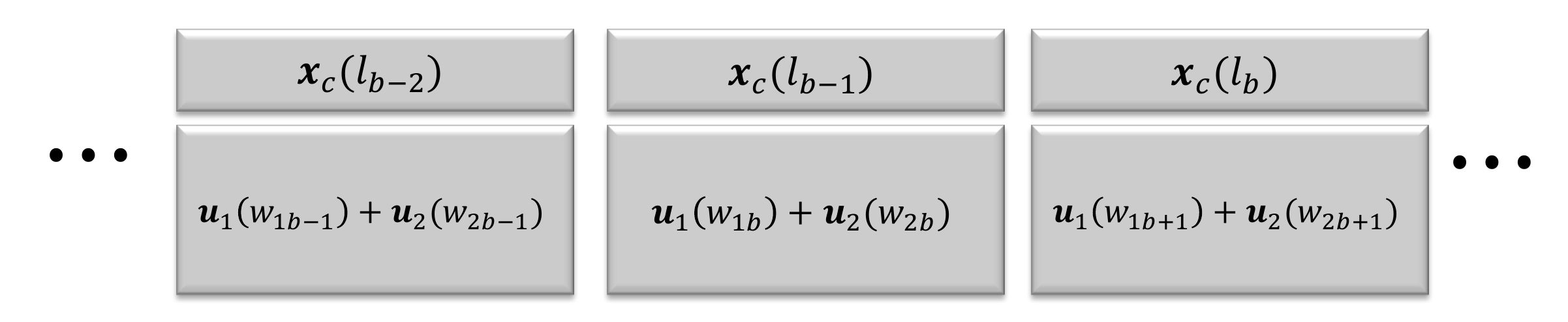}
\caption{Block-Markov Encoding.}
\label{fig:BME}
\end{figure}

The transmission consists of $B$ blocks of length~$n$. For simplicity of
demonstration, we set $n=1$. 
In block $b$, the
transmitter sends a mixture of two new private messages $w_{1b}$ and
$w_{2b}$ together with one common message $l_{b-1}$, for $b=1,\ldots,B$.
As it will become clear, the message $l_{b-1}$ is the
compression index of the mutual interferences experienced by the
receivers in the previous block~$b-1$. 
By encoding $w_{1b}$, $w_{2b}$, and $l_{b-1}$ into $\uv_1(w_{1b})$, $\uv_2(w_{2b})$, and $\xv_c(l_{b-1})$, respectively,
 with independent channel codebooks,  
 the transmitted signal is
written as
\begin{align} \label{eq:tx-signal-211}
  \xv[b] &= \xv_c(l_{b-1}) + \uv_1(w_{1b}) + \uv_2(w_{2b}), \quad
  b=1,\ldots,B
\end{align}%
where we set $l_{0} = 1$ to initiate the transmission and $w_{1B} =
w_{2B} = 1$ to end it. As before, the common message $\xv_c(l_{b-1})$ is with power $P$, whereas the precoding in $\uv_1$ and $\uv_2$ is with a reduced power, parameterized
by $A$, $A'$, with $0 \le A,A'\le 1$, i.e., 
\begin{align} \label{eq:cov-uv-211}
  \Qm_{1} &= P^{A} \Phim_{\hat{h}_2} + P^{A'}
  \Phim_{\hat{h}_2^\bot},\quad \Qm_{2} = P^{A}
  \Phim_{\hat{h}_1} + P^{A'} \Phim_{\hat{h}_1^\bot}
\end{align}%
where $A \defeq (A'-\alpha)^+$.
The mutual interferences are defined similarly and their powers are now
reduced
\begin{align}
  y_1[b] &= \underbrace{\hv_1^\H \xv_c(l_{b-1})}_{P} + \underbrace{\hv_1^\H
  \uv_1(w_{1b})}_{P^{A'}} + \underbrace{\hv_1^\H \uv_2(w_{2b})}_{\eta_{1b}\sim P^{A}}  \\
  y_2[b] &= \underbrace{\hv_2^\H \xv_c(l_{b-1})}_{P} +
  \underbrace{\hv_2^\H \uv_2(w_{2b})}_{P^{A'}} + \underbrace{\hv_2^\H
  \uv_1(w_{1b})}_{\eta_{2b}\sim P^{A}}  
\end{align}%
where we omit the block indices for the channel coefficients as well as
the AWGN for brevity. At the end of block~$b$, $(\eta_{1b},\eta_{2b})$
are compressed with a codebook of size $P^{2A}$ into an index
$l_b\in \left\{ 1,\ldots,P^{2A} \right\}$. The distortion
between $(\eta_{1b}, \eta_{2b})$ and $(\hat{\eta}_{1}(l_b),
\hat{\eta}_{2}(l_b))$ is at the noise level. 
 
At the end of $B$ blocks, Receiver~$k$ would like to retrieve $w_{k1},
\ldots, w_{k,B-1}$. Let us focus on Receiver~1, without loss of
generality. In this particular case, $l_{b-1}$ can
be decoded at the end of block $b$, by treating the private signals as
noise, i.e., with signal-to-interference-and-noise-ratio~(SINR)
level~$P^{1-A'}$, for $b=2,\ldots,B$. The correct decoding of $l_{b-1}$ is guaranteed if the
SINR can support the DoF of $2A$ for the common message, i.e., 
\begin{align}\label{eq:dec-eta-211}
   2A \le 1-A'. 
\end{align}%
Given that this condition is satisfied, $l_0, l_1, \ldots, l_{B-1}$ are
available to both receivers. Therefore, ${\eta_{1b}, \eta_{2b}}$,
$b=1,\ldots,B-1$, are known, up to the noise level. 
To decode $w_{1b}$, Receiver~1 uses $\eta_{1b}$, $\eta_{2b}$, and
$l_{b-1}$ to form the following $2\times 2$ MIMO system
\begin{align}
  \begin{bmatrix} y_1[b] - \hv_1^\H \xv_c(l_{b-1}) - \eta_{1b} \\ \eta_{2b}
  \end{bmatrix} &= \begin{bmatrix} \hv_1^\H \\ \hv_2^\H \end{bmatrix}
    \uv_1(w_{1b})
\end{align}%
where the equivalent channel matrix has rank~$2$ almost surely. 
This decoding strategy for the private message can be boiled down into the backward decoding, where the mutual interferences (cf.~$\eta_{1b}$, $\eta_{2b}$) decoded in the future block are utilized in current block as side information.
From the covariance matrix $\Qm_1$ of $\uv_1$ from \eqref{eq:cov-uv-211},
we deduce that the correct decoding of $w_{1b}$ is guaranteed if the
DoF $d_{1b}$ of $w_{1b}$ satisfies
\begin{align}\label{eq:dec-w-211}
  d_{1b} &\le A + A'. 
\end{align}%
Combining \eqref{eq:dec-eta-211} and \eqref{eq:dec-w-211}, it readily
follows that the optimal $A'$ should equalize \eqref{eq:dec-eta-211},
i.e., $A'^* = \frac{1+2\alpha}{3}$. Thus, we achieve $d_{1b} =
\frac{2+\alpha}{3}$. Due to the symmetry, $d_{2b}$ has the same value.
Finally, we have
\begin{align}
  d_k &= \frac{1}{B} \sum_{b=1}^{B-1} d_{kb} =
  \frac{B-1}{B}\frac{2+\alpha}{3}, \quad k=1,2
\end{align}%
which goes to $\frac{2+\alpha}{3}$ when $B\to\infty$. 

By now, we have shown that   
both MAT and $\alpha$-MAT schemes can be interpreted under a common framework of block-Markov encoding with power allocation parameters $(A, A')$ and that they only differ from the choice of these parameters. As we will show in the following subsections, the strength (or benefit) of the block-Markov encoding framework becomes evident in the asymmetric system setting, where the original $\alpha$-MAT alignment fails to achieve the optimal DoF in general.

\subsection{Asymmetry in Current CSIT Qualities}
Let us consider again the MISO BC case but assume now that the CSIT qualities of two channels are different, 
i.e., $\alpha_1\neq \alpha_2$,  
where $\alpha_k$
$(k=1,2)$ for Receiver~$k$. 
The signal model is in the exact same form as in
\eqref{eq:tx-signal-211} with a more general precoding, parameterized
by $A_k$, $A_k'$, with $0 \le A_k, A_k'\le1$:
\begin{align} \label{eq:cov-uv-211-2}
  \Qm_{1} &= P^{A_1} \Phim_{\hat{h}_2} + P^{A_1'}
  \Phim_{\hat{h}_2^\bot},\quad \Qm_{2} = P^{A_2}
  \Phim_{\hat{h}_1}
  + P^{A_2'} \Phim_{\hat{h}_1^\bot}
\end{align}%
where $A_k \defeq (A_k'-\alpha_j)^+$, $j \ne k \in \{1,2\}$.
Following the same footsteps as in the symmetric case, it is readily
shown that $\eta_{1b} \sim P^{A_2}$ and $\eta_{2b}\sim
P^{A_1}$ and that $(\eta_{1b}, \eta_{2b})$ can be compressed
up to the noise level with a source codebook of size
$P^{A_1 + A_2}$. The decoding at both
receivers is the same as before. To decode the common message
$l_{b-1}$ by treating the private signals as noise, since the SINR is
$P^{1-A_1'}$ at Receiver~1 and $P^{1-A_2'}$ at Receiver~2, the DoF of the
common message should satisfy 
\begin{align}
A_1 + A_2 &\le \min\{1 - A_1', 1 - A_2'\}. 
\end{align}
Using the common messages $l_{b}$ and $l_{b-1}$ as side information,
$w_{1b}$ and $w_{2b}$ can be decoded at the respective receivers if 
\begin{align}
  d_{1b} &\le A_1 + A_1' \quad \text{and} \quad
  d_{2b} \le A_2 + A_2'. 
\end{align}%
By carefully selecting the parameters $A_1'$ and $A_2'$, all corner points
of the DoF outer bound can be achieved, as shown in the following table. 
Note that the condition is to make sure the corner points exist, in which the corner point $(\alpha_2,1)$ always exists as long as $\alpha_1 \ge \alpha_2$.
\begin{table}[htbp]
 \caption{Parameter Setting for the $(2,1,1)$ BC Case $(\alpha_1 \ge \alpha_2)$}
  \centering
  \begin{tabular}{@{} c|c|c|c @{}} \hline  \hline
      Condition & $A_1'$ & $A_2'$ & Corner Point $(d_1, d_2)$  \\ \hline 
      \multirow{2}{*}{ $2\alpha_1 - \alpha_2 \le 1$}&
      $A_1'=\frac{1+\alpha_1+\alpha_2}{3}$ & $A_2'=\frac{1+\alpha_1+\alpha_2}{3} $ & $\left(\frac{2+2\alpha_1-\alpha_2}{3}, \frac{2-\alpha_1+2\alpha_2}{3}\right)$ \\ \cline{2-4}
      &  $A_1'=\frac{1+\alpha_2}{2}$ &   $A_2'=\alpha_1$ & $(1,\alpha_1)$ \\ \hline 
      \multirow{1}{*}{ $2\alpha_1 - \alpha_2 > 1$}   &
      $A_1'=\frac{1+\alpha_2}{2}$ & $A_2'=\frac{1+\alpha_2}{2} $ &
      $(1,\frac{1+\alpha_2}{2})$ \\ \cline{1-4}
      - &  $A_1'=\alpha_2$ & $A_2'=\frac{1+\alpha_1}{2}$ & $(\alpha_2,1)$ \\  \hline \hline
\end{tabular}
  \label{tab:bc_ex_211}
\end{table}

\subsection{Asymmetry in Antenna Configurations}

We use the $(4,3,2)$ MIMO BC case to show that the block-Markov encoding
can achieve the optimal performance in asymmetric antenna setting. 
Recall that, in the previous subsections, the backward decoding is performed to decode
the private messages, and that the common messages can be decoded {\em block by block}.
In this
case, nevertheless, we also need backward decoding to decode the common messages as well.  

The same transmission signal model \eqref{eq:tx-signal-211} is used here, with
the following precoding, parameterized by $A_k$ and $A_k'$, $(k=1,2)$,
$0 \le A_k \le A_k' \le1$:
\begin{align} \label{eq:cov-uv-432}
  \Qm_{1} &= P^{A_1} \Phim_{\hat{H}_2} + P^{A'_1}
  \Phim_{\hat{H}_2^\bot},\quad \Qm_{2} = P^{A_2}
  \Phim_{\hat{H}_1} + P^{A'_2} \Phim_{\hat{H}_1^\bot} 
\end{align}%
where $A_k$, $k \ne j \in \{1,2\}$, is defined as
\begin{align} \label{eq:Ak-def-ex}
A_k \defeq \left\{ \Pmatrix{(A'_k-\alpha_j)^+, & d_k \le 4-N_j\alpha_j, \\ \frac{d_k-(4-N_j)}{N_j}, & d_k > 4-N_j\alpha_j}   \right.
\end{align}
with $d_k \in \mathbb{R}_+$ being the achievable DoF associated with Receiver~$k$.
It is readily verified that $ A'_k-\alpha_j \le A_k \le A_k'$ is always true, such
that the created interference at intended Receiver~$j$ is of power level $A_k$, and the desired signal at Receiver~$k$ is of level $A_k'$.

We recall that the common message $\xv_c(l_{b-1})$ is transmitted with power $P$ and that the ranks of $\Phim_{\hat{H}_2}$,
$\Phim_{\hat{H}_2^\bot}$, $\Phim_{\hat{H}_1}$, and $\Phim_{\hat{H}_1^\bot}$
are respectively $2$, $2$, $3$, and $1$, almost surely. The received
signals are now vectors given by
\begin{align}
  \yv_1[b] &= \underbrace{\Hm_1 \xv_c(l_{b-1})}_{P \Id_3} + {\Hm_1
  \uv_1(w_{1b})} + \underbrace{\Hm_1 \uv_2(w_{2b})}_{\etav_{1b} \sim
  P^{A_2} \Id_3},  \\
  \yv_2[b] &= \underbrace{\Hm_2 \xv_c(l_{b-1})}_{P \Id_2} + {\Hm_2
  \uv_2(w_{2b})} + \underbrace{\Hm_2 \uv_1(w_{1b})}_{\etav_{2b} \sim
  P^{A_1} \Id_2}  .
\end{align}%
Following the same footsteps as in the single receive antenna case, it
is readily shown that 
$(\etav_{1b}, \etav_{2b})$ can be compressed up to the noise level with a
source codebook of size $P^{2A_1 + 3A_2}$. For convenience, let us define 
\begin{align}
  d_{\eta} \defeq {2 A_1 + 3 A_2}.
  \label{eq:deta-432}
\end{align}%

Unlike the MISO case where the common messages are decoded in each block
independently, backward decoding is required for both common and private
messages. As we will see later on, the common rate can be improved with
backward decoding in general. The decoding starts at block $B$, since
$w_{1B}$ and $w_{2B}$ are both known, the private signals can be removed
from the received signals $\yv_1[B]$ and $\yv_2[B]$. The common message
$l_{B-1}$ can be decoded at both receivers if $d_{\eta} \le 2$. At block
$b$, for $b=B-1,\ldots,2$, assuming $l_b$ is known perfectly from the
decoding of block~$b+1$, $\etav_{1b}$ and $\etav_{2b}$ can be
reconstructed up to the noise level. The following MIMO system can be
obtained 
\begin{align}
  \begin{bmatrix} \yv_1[b] - \etav_{1b} \\ \etav_{2b}
  \end{bmatrix} &= \begin{bmatrix} \Hm_1 \\ 0 \end{bmatrix}
    \xv_c(l_{b-1}) +
\begin{bmatrix} \Hm_1 \\ \Hm_2 \end{bmatrix}
    \uv_1(w_{1b}).
\end{align}%
Note that this is a multiple-access channel~(MAC) from which
$l_{b-1}$ and $w_{1b}$ can be correctly decoded if the rate pair lies
within the following region
\begin{align}
  d_{\eta} &\le 3 \label{eq:MAC-deta1-432}\\
  d_{1b} &\le 2 A_1 + 2 A_1' \label{eq:MAC-d1b-432}\\
  d_{\eta} + d_{1b} &\le 3 + 2 A_1, \label{eq:MAC-dsum1-432}
\end{align}
whose general proof is provided in Appendix~\ref{app:MAC}.
Let us set $d_{1b}$ to equalize \eqref{eq:MAC-d1b-432}. Then,
\eqref{eq:MAC-deta1-432} and \eqref{eq:MAC-dsum1-432} imply $d_{\eta}
\le 3-2 A'_1$. Similar analysis on Receiver~2 will lead to $d_{\eta}
\le 2- A'_2$, by setting $d_{2b} = A'_2 + 3 A_2$. Therefore,
from \eqref{eq:deta-432}, we obtain the following constraint
\begin{align}
  {2A_1 + 3A_2} \le \min\left\{ 3-2A'_1, 2-A'_2 \right\}   \label{eq:deta-432-constraint}
\end{align}%
to achieve any $(d_{1b}, d_{2b})$ such that
\begin{align}
  d_{1b} \le 2 A_1 + 2 A_1' \quad \text{and} \quad
  d_{2b} \le A'_2 + 3 A_2 . 
\end{align}%
By letting $B\to\infty$, $d_{1} = 2 A_1 + 2 A_1'$ and
$d_2 = A'_2 + 3 A_2$ can be achieved for any $A'_1, A'_2 \le
1$ given the definition of $(A_1,A_2)$ in \eqref{eq:Ak-def-ex}, as long as \eqref{eq:deta-432-constraint} is satisfied. We can show
that, by properly choosing $(A'_1, A'_2)$, all the corner points given
by the outer bound can be achieved. For example, by setting
$\alpha_1=\alpha_2=\alpha$, the values $(A'_1,A'_2)$ that achieve the
corner points are illustrated in Tab.~\ref{tab:bc_ex_432}.  Note that $(\frac{12}{5}, \frac{4}{5} + \alpha)$ exists when $\alpha \le \frac{4}{5}$, whereas $(3\alpha, 4-3\alpha)$ and $(4-2\alpha, 2\alpha)$ exist when $\alpha > \frac{4}{5}$.
\begin{table}[htbp]
 \caption{Parameter Setting for the $(4,3,2)$ BC Case with
 $\alpha_1=\alpha_2=\alpha$. }
  \centering
  \begin{tabular}{@{} c|c|c|c|c @{}} \hline  \hline
  Corner Point $(d_1, d_2)$ & Cond. & $(A'_1, A'_2)$ & $(A_1,A_2)$ & $d_{\eta}$  \\ \hline \hline      
  \multirow{2}{*}{$(3, \alpha)$} & $\alpha \le \frac{1}{2}$ & $(\frac{3+2\alpha}{4},\alpha)$ & $(\frac{3-2\alpha}{4},0)$ & $\frac{3-2\alpha}{2}$ \\ \cline{2-5}
  & $\alpha > \frac{1}{2}$ & $(1,\alpha)$ & $(\frac{1}{2},0)$ & 1 \\ \hline
  \multirow{2}{*}{$(2\alpha, 2)$} & $\alpha \le \frac{2}{3}$ & $(\alpha,\frac{2+3\alpha}{4})$ & $(0,\frac{2-\alpha}{4})$ & $\frac{6-3\alpha}{4}$  \\  \cline{2-5}
   & $\alpha > \frac{2}{3}$ & $(\alpha,1)$ & $(0,\frac{1}{3})$ & 1  \\  \hline
  $(\frac{12}{5}, \frac{4}{5} + \alpha)$ & $\alpha \le \frac{4}{5}$ & $(\frac{3}{5}+\frac{1}{2}\alpha,\frac{1}{5}+\alpha)$ & $(\frac{3}{5}-\frac{1}{2}\alpha,\frac{1}{5})$& $\frac{9}{5}-\alpha$  \\  \hline 
  $(3\alpha, 4-3\alpha)$ & $\alpha > \frac{4}{5} $ & $(1,1)$ & $(\frac{3\alpha-2}{2},1-\alpha)$ & 1  \\ \hline
  $(4-2\alpha, 2\alpha)$ & $\alpha > \frac{4}{5} $ & $(1,1)$ & $(1-\alpha,\frac{2\alpha-1}{3})$ & 1  \\ \hline
  \hline
  \end{tabular}
  \label{tab:bc_ex_432}
\end{table}

\section{Achievability: the General Formulation}
As aforementioned, the key ingredients of the achievability scheme consist of:
\begin{itemize}
\item {\em block-Markov encoding with a constant block length}: the
  fresh messages in the current block and the interference created in
  the past blocks are encoded together with the proper rate splitting and power scaling;
 \item {\em spatial precoding with imperfect current CSIT}: with proper
   power allocation on the symbols that delivered via the range and null
   spaces of the inaccurate current channel, the interference at unintended receiver can be reduced compared to that without any CSIT;
\item {\em interference quantization}: instead of forwarding the
  overhead interference directly in an analog way as done in pure delayed CSIT scenario, the reduced-power interferences are compressed first with reduced number of bits, and forwarded in a digital fashion with lower rate;
\item {\em backward decoding}:  the messages are decoded from the last block to the first one, where in each block the messages are decoded with the aid of side information provided by the blocks in the future. 
\end{itemize}

In the following, the general achievability scheme will be described in detail for BC and IC respectively.

\subsection{Broadcast Channels}
First of all, we notice that the region \eqref{eq:outer-bound} given in
Theorem~\ref{theorem:BC} does not depend on $M$~(resp. $N_k$) when
$M> N_1+N_2$~(resp. $N_k> M$). Therefore, it is sufficient to prove the
achievability for the case $M\le N_1+N_2$ and $N_k \le M$. And the
achievability for the other cases can be inferred by simply switching off 
the additional transmit/receive antennas. Thus, it yields
\begin{align} \label{eq:bc-ant-conf}
  M = \min\left\{ M, N_1+N_2 \right\}, \quad N_k = \min\left\{
  M, N_k \right\}, \ k=1,2. 
\end{align}

\subsection*{\underline{Block-Markov encoding}}

The block-Markov encoding has the same structure as before, given by 
\begin{align} \label{eq:tx-signal}
  \xv[b] &= \xv_c(l_{b-1}) + \uv_1(w_{1b}) + \uv_2(w_{2b}), \quad
  b=1,\ldots,B
\end{align}%
where we recall that we set $l_{0} = 1$ to initiate the transmission and $w_{1B} =
w_{2B} = 1$ to end it. 

\subsection*{\underline{Spatial precoding}}

Both $\uv_1, \uv_2\in\CC^{M\times 1}$ are precoded signals of $M$
streams, i.e.,   
\begin{align} \label{eq:cov-uv}
  \Qm_{1} &= P^{A_{1}} \Phim_{\hat{H}_2} + P^{A'_{1}}
  \Phim_{\hat{H}_2^\bot},\quad \Qm_{2} = P^{A_{2}}
  \Phim_{\hat{H}_1} + P^{A'_{2}} \Phim_{\hat{H}_1^\bot} 
\end{align}%
where the rank of $\Phim_{\hat{H}_k}$ is $N_k$ whereas
the rank of $\Phim_{\hat{H}_k^\bot}$ is $M-N_k$, $k=1,2$. 
In other words, for Receiver~$k$, $N_j$ streams are sent in the subspace of the
unintended Receiver~$j$ with power level $A_k$ and the other $M-N_j$
streams are sent in the null space of Receiver~$j$ with power level $A'_k$, 
where $(A_k, A_k')$ satisfies
\begin{align} \label{eq:def-Ak-bc}
  0\le A_k \le A_k' \le 1 \quad \text{and} \quad
  A_k \ge A'_k - \alpha_j
\end{align}%
for $j \neq k \in \left\{ 1,2 \right\}$. 
Note that the above condition guarantees that the interference
at Receiver~$j$ has power level $A_k$ and the desired signal at 
Receiver~$k$ is of power level $A_k'$.

\subsection*{\underline{Interference quantization}}

Recall that the common message $\xv_c(l_{b-1})$ is sent with power $P$. The received signals in block $b$ are given by 
\begin{align}
  \yv_1[b] &= \underbrace{\Hm_1 \xv_c(l_{b-1})}_{P \Id_{N_1}} + {\Hm_1
  \uv_1(w_{1b})} + \underbrace{\Hm_1 \uv_2(w_{2b})}_{\etav_{1b} \sim
  P^{A_2} \Id_{N_1}} , \\
  \yv_2[b] &= \underbrace{\Hm_2 \xv_c(l_{b-1})}_{P \Id_{N_2}} + {\Hm_2
  \uv_2(w_{2b})} + \underbrace{\Hm_2 \uv_1(w_{1b})}_{\etav_{2b} \sim
  P^{A_1} \Id_{N_2}}  .
\end{align}%
It is readily shown that 
$(\etav_{1b}, \etav_{2b})$ can be compressed up to the noise level with a
source codebook of size $P^{N_2 A_1 +
N_1 A_2}$ into an index $l_b$. For convenience, let us define 
\begin{align}
  d_{\eta_1} \defeq {N_1 A_2}, \ 
  d_{\eta_2} \defeq {N_2 A_1}, \ 
  \text{and } d_{\eta} \defeq d_{\eta_1} + d_{\eta_2}. 
  \label{eq:deta}
\end{align}%

\subsection*{\underline{Backward decoding}}

The decoding starts at block $B$, since
$w_{1B}$ and $w_{2B}$ are both known, the private signals can be removed
from the received signals $\yv_1[B]$ and $\yv_2[B]$. The common message
$l_{B-1}$ can be decoded at both receiver if $d_{\eta} \le \min\left\{
N_1, N_2 \right\}$. At block
$b$, assuming $l_b$ is known perfectly from the
decoding of block~$b+1$, $\etav_{1b}$ and $\etav_{2b}$ can be
reconstructed up to the noise level, for $b=B-1,\ldots,2$. The following MIMO system can be
obtained at Receiver~$k$, $k=1,2$
\begin{align} \label{eq:BC-MAC}
  \begin{bmatrix} \yv_k[b] - \etav_{kb} \\ \etav_{jb}
  \end{bmatrix} &= \begin{bmatrix} \Hm_k \\ 0 \end{bmatrix}
    \xv_c(l_{b-1}) +
\begin{bmatrix} \Hm_k \\ \Hm_j \end{bmatrix}
    \uv_k(w_{kb})
\end{align}%
for $j\ne k \in \left\{ 1, 2 \right\}$. 
Since the common message $l_{b-1}$ and the private message $w_{kb}$ are both desired by Receiver~$k$,
this system corresponds to a multiple-access channel~(MAC). As formally proved in Appendix~\ref{app:MAC}, Receiver~$k$
can decode correctly both messages if the following conditions are satisfied.
\begin{align}
  d_{\eta} &\le N_k \label{eq:MAC-detak}\\
  d_{kb} &\le N_j A_k + (M-N_j) A'_k  \label{eq:MAC-dkb}\\
  d_{\eta} + d_{kb} &\le N_k + N_j A_k. \label{eq:MAC-dsum}
\end{align}

Let us choose $d_{kb}$ to be equal to the right hand side of \eqref{eq:MAC-dkb} for $k=1, 2$ and $b=1, .., B-1$.
Then, the equality in \eqref{eq:MAC-dkb} together with \eqref{eq:deta}, \eqref{eq:MAC-detak}, \eqref{eq:MAC-dsum} implies when letting $B\to \infty$ the following lemma.

\begin{lemma}[decodability condition for BC] \label{lemma:decodability-BC}
  Let us define 
  \begin{align} 
    \AsetBC &\defeq \bigl\{ (A_1, A'_1, A_2, A'_2)\ \vert \ A_k, A'_k \in
    [0,1],\ A'_k - \alpha_j \le A_k \le A'_k, \quad \forall\, k\ne j \in
    \left\{ 1,2 \right\} \bigr\} \label{eq:Aset-BC} \\
    \DsetBC &\defeq \left\{ (d_1, d_2) \ \vert \ d_k \in [0,N_k], \quad
    \forall\,k \in \left\{ 1,2 \right\}  \right\}\\
    \shortintertext{and}
    \fAd \ : \ \AsetBC &\to \DsetBC \\
    (A_k, A'_k) &\mapsto d_k \defeq N_j A_k + (M-N_j) A'_k, \quad
    \forall k\ne j \in\{1,2\} . \label{eq:d1d2}
  \end{align}%
  Then $(d_1,d_2) = \fAd(\Avec)$, for some $\Am \in \AsetBC$, is achievable
  with the proposed scheme, if 
\begin{align}
  d_{\eta_1} + d_{1} &\le N_1, \label{eq:interp1} \\ 
  d_{\eta_2} + d_{2} &\le N_2. \label{eq:interp2}  
\end{align}%
where we recall $d_{\eta_1} \defeq N_1 A_2$ and $d_{\eta_2} \defeq N_2 A_1$. 
\end{lemma}

\begin{remark}
  In the above lemma, $d_{\eta_k}$ can be interpreted as the degrees of freedom
  occupied by the interference 
  at Receiver~$k$. Therefore, \eqref{eq:interp1} and \eqref{eq:interp2}
  are clearly outer bounds for any transmission strategies, i.e., the
  sum of the dimension of the useful signal and the dimension of the
  interference signal at the receiver side cannot exceed the total
  dimension of the signal space. These bounds are in general not tight
  except for special cases such as the ``strong interference'' regime
  where interference can be decoded completely and removed or the ``weak
  interference'' regime where the interference can be treated as noise
  while the useful signal power dominates the received power.
  Remarkably, the proposed scheme achieves these outer bounds. This is
  due to two of the main ingredients of our scheme, namely, interference 
  quantization and the block-Markov encoding. The block-Markov encoding 
  places the digitized interference in the ``upper level'' of the signal
  space~(with full power) and thus ``pushes'' the channel into the ``strong interference''
  regime in which the digitized interference can be decoded thanks to 
  the structure brought by the interference quantization.  
\end{remark}

\begin{definition}[achievable region for BC]
  Let us define
  \begin{align} \label{eq:IABC}
    \IABC &\defeq  \left\{ (A_1, A'_1, A_2, A'_2) \in \AsetBC \left| \Pmatrix{
    (d_1, d_2) = \fAd(A_1, A'_1, A_2, A'_2), \\
    \frac{d_k}{N_k} \le 1 - A_j, \quad k\ne j\in\left\{ 1, 2 \right\}} \right. \right\}
  \end{align}%
and the achievable DoF region of the proposed scheme
  \begin{align} \label{eq:IdBC}
    \IdBC &\defeq \fAd(\IABC) \defeq 
    \left\{ (d_1, d_2) \left| \Pmatrix{
    (d_1, d_2) = \fAd(A_1, A'_1, A_2, A'_2), \\
    (A_1, A'_1, A_2, A'_2) \in \AsetBC, \\
    \frac{d_k}{N_k} \le 1 - A_j, \quad k\ne j\in\left\{ 1, 2 \right\}} \right. \right\}.
  \end{align}%
\end{definition}

\subsection*{\underline{Achievability analysis}}

In the following, we would like to show that any pair~$(d_1, d_2)$ in
the outer bound region defined by \eqref{eq:outer-bound}, hereafter
referred to as $\OdBC$, can be achieved by the proposed strategy.
Therefore, it is sufficient to show that $\OdBC \subseteq \IdBC$. The
main idea is as follows. If there exists a function 
\begin{align}
  \fdA \ &: \ \OdBC \to \AsetBC \\
  \shortintertext{such that}
  (d_1, d_2) &= \fAd(\fdA(d_1, d_2)),\quad \text{and}
  \label{eq:inverse-BC} \\
  \fdA(d_1, d_2) &\in \IABC, \label{eq:cond-dec-BC}
\end{align}%
then for every $(d_1,d_2)\in \OdBC$ we can use the power allocation
$(A_1,A'_1,A_2,A'_2) = \fdA(d_1, d_2)$ on the proposed scheme to achieve
it, i.e., 
\begin{align}
  \OdBC = \fAd(\fdA(\OdBC)) \subseteq \fAd(\IABC) = \IdBC
\end{align}%
from which the achievability is proved. 
Now, we define formally the power allocation function. 
\begin{definition}[power allocation for BC] \label{def:Ak-dk-bc}
Let us define $\fdA\ : \ \OdBC \to \AsetBC$:
\begin{align}
  (d_1, d_2) \mapsto (A_1, A'_1) \defeq f_1(d_1),\ (A_2, A'_2) \defeq
  f_2(d_2)
\end{align}%
where $f_k$, $j\ne k \in \{1,2\}$, is specified as below.
\begin{itemize}
  \item When $M=N_j$: $A'_k=A_k = \frac{d_k}{M}$;
  \item When $M>N_j$ and $d_k < M-N_j \alpha_j$: $A_k = (A'_k -
    \alpha_j)^+$, and thus
    \begin{align}
      A'_k &= \begin{cases} 
        \frac{d_k}{M-N_j}, & \text{if } d_k < (M-N_j) \alpha_j;  \\
        \frac{d_k + N_j\alpha_j}{M}, & \text{otherwise};
      \end{cases}
    \end{align}%
  \item When $M>N_j$ and $d_k \ge M-N_j \alpha_j$: $A'_k=1$, and thus
    $A_k=\frac{d_k - (M-N_j)}{N_j}$. 
\end{itemize}
\end{definition}
It is readily shown that, for any $(d_1,d_2) \in \OdBC$,  the resulting
power allocation always lies in $\AsetBC$ as defined in \eqref{eq:Aset-BC}
and that \eqref{eq:inverse-BC} is always satisfied. It remains to show
that \eqref{eq:cond-dec-BC} holds as well, i.e., the decodability
condition in \eqref{eq:IABC} is satisfied. To that end, for any $(d_1, d_2) \in \OdBC$, 
we first define $(A_1,A'_1,A_2,A_2') \defeq \fdA(d_1,d_2)$ which implies
$d_j =  N_k A_j + (M-N_k) A'_j$, $j\ne k \in \{1,2\}$. Applying this equality
on the constraints in the outer bound $\OdBC$ in
\eqref{eq:outer-bound}, we have  
\begin{align}
  \frac{d_k}{N_k} &\le \frac{M-(M-N_k)A'_j}{N_k} - A_j, \label{eq:tmp666}\\
   \frac{d_k}{N_k} &\le 1 - \left[ \frac{(M-N_k)(A'_j-\alpha_k) +
   N_k A_j }{M}  \right]^+ \label{eq:tmp667}
\end{align}%
for $k\ne j\in\left\{ 1, 2 \right\}$, where the first one is from the
sum rate constraint \eqref{eq:bc-outer-bound-3} whereas the second one
is from the rest of the constraints in \eqref{eq:outer-bound}. 
The final step is to show that either of \eqref{eq:tmp666} and
\eqref{eq:tmp667} implies the last constraint in \eqref{eq:IABC}:
\begin{itemize}
  \item When $M=N_k$, \eqref{eq:tmp667} is identical to the last
    constraint in \eqref{eq:IABC}; 
  \item When $M>N_k$ and $d_j \ge M - N_k \alpha_k$, we have $A'_j=1$
    according to the map $\fdA$ defined in Definition~\ref{def:Ak-dk-bc}. Hence, 
    \eqref{eq:tmp666} is identical to the last
    constraint in \eqref{eq:IABC}; 
  \item When $M > N_k$ and $d_j < M - N_k \alpha_k$, we have $A_j =
    (A'_j - \alpha_k)^+$ according to Definition~\ref{def:Ak-dk-bc}. Hence, 
    \begin{align} \label{eq:const-bc-1}
      \left[ \frac{(M-N_k)(A'_j-\alpha_k) +
      N_k A_j }{M}  \right]^+ &\ge A_j 
    \end{align}%
    with which \eqref{eq:tmp667} implies the last constraint in
    \eqref{eq:IABC}. 
\end{itemize}
By now, we have proved the achievability through the existence of a
proper power allocation function such that \eqref{eq:inverse-BC} and
\eqref{eq:cond-dec-BC} are satisfied for every pair $(d_1, d_2)$ in the
outer bound.

\subsection{Interference Channels}

The proposed scheme for MIMO IC is similar to that for BC, with the
differences that (a) the interferences can only be reconstructed at the
transmitter from which the symbols are sent, and (b) antenna configuration does
matter at both transmitters and receivers. 
Further, as with the broadcast channel, we notice that the region
\eqref{eq:outer-bound-ic} given in Theorem~\ref{theorem:IC} does not
depend on $M_k$~(resp. $N_k$) when $M_k> N_1+N_2$~(resp. $N_k>
M_1+M_2$), $k=1,2$. Therefore, it is sufficient to prove the achievability
for the case $M_k\le N_1+N_2$ and $N_k \le M_1 + M_2$, $k=1,2$, since the achievability for
the other cases can be inferred by simply switching off the additional
transmit/receive antennas. Thus, it yields
\begin{align} \label{eq:ic-ant-conf}
  M_k = \min\left\{ M_k, N_1+N_2 \right\}, \quad N_k = \min\left\{ N_k, M_1+M_2 \right\},
  \quad k=1,2. 
\end{align}
We also define for notational convenience
\begin{align}
  N'_1 &\defeq \min\left\{ N_1, M_2 \right\}, \quad
  N'_2 \defeq \min\left\{ N_2, M_1 \right\}.
\end{align}%

\subsection*{\underline{Block-Markov encoding}}

The block-Markov encoding is done independently at both transmitters
\begin{align} \label{eq:tx-signal-ic}
  \xv_1[b] &= \xv_{1c}(l_{1,b-1}) + \uv_1(w_{1b}), \\
  \xv_2[b] &= \xv_{2c}(l_{2,b-1}) + \uv_2(w_{2b}), \quad b=1,\ldots,B
\end{align}%
where we set $l_{1,0} = l_{2,0} = 1$ to initiate the transmission and $w_{1B} =
w_{2B} = 1$ to end it. 

\subsection*{\underline{Spatial precoding}}

The signal $\uv_k \in \CC^{M_k\times 1}$, $k=1,2$, is precoded signal of
$M_k$ streams, i.e.,   
\begin{align} \label{eq:cov-uv-ic}
  \Qm_{1} &= P^{A_{1}} \Phim_{\hat{H}_{21}} + P^{A'_1} \Phim_{\hat{H}^{\perp1}_{21}} + P^{A''_{1}} \Phim_{\hat{H}^{\perp2}_{21}},\\ 
  \Qm_{2} &= P^{A_{2}} \Phim_{\hat{H}_{12}} + P^{A'_2} \Phim_{\hat{H}^{\perp1}_{12}} + P^{A''_{2}} \Phim_{\hat{H}^{\perp2}_{12}} 
\end{align}%
where we use
$\hat{\Hm}^{\perp1}_{jk}$~(resp.~$\hat{\Hm}^{\perp2}_{jk}$) to denote any matrix spanning the $(M_k - N'_j -
\xi_k)$-dimensional~(resp.~$\xi_k$-dimensional) subspace of the null
space of $\hat{\Hm}_{jk}$ where $\xi_k$ will be specified later on. 
Therefore, the rank of $\Phim_{\hat{H}_{jk}}$ is $N'_j$ whereas
the rank of $\Phim_{\hat{H}_{jk}^{\perp1}}$ and
$\Phim_{\hat{H}_{jk}^{\perp2}}$ are respectively $M_k-N'_j-\xi_k$ and
$\xi_k$, $k=1,2$. The power levels $(A_k, A'_k, A''_k)$ satisfy
\begin{align} \label{eq:const-Ak-ic}
  A_k, A'_k, A''_k \in [0, 1], \ A_k \le A'_k, \ A''_k \le A'_k, \quad \text{and} \quad
  A_k \ge A'_k - \alpha_j
\end{align}%
for $j \ne k \in \left\{ 1, 2 \right\}$. 
Note that the above condition guarantees that the interference
at Receiver $j$ has power level $A_k$ and the desired signal at Receiver~$k$ at power level $A_k'$. 

\subsection*{\underline{Interference quantization}}

Recall that the common messages $\xv_{1c}(l_{1,b-1})$ and $\xv_{2c}(l_{2,b-1})$ are sent with power $P$.  The received signals in block $b$ are given by
\begin{align}
  \yv_1[b] &= \underbrace{\Hm_{11} \xv_{1c}(l_{1,b-1}) + \Hm_{12}
  \xv_{2c}(l_{2,b-1})}_{P \Id_{N_1}} + {\Hm_{11} \uv_1(w_{1b})} + \underbrace{\Hm_{12} \uv_2(w_{2b})}_{\etav_{1b} \sim P^{A_2} \Id_{N'_1}},  \\
  \yv_2[b] &= \underbrace{\Hm_{22} \xv_{2c}(l_{2,b-1}) + \Hm_{21}
  \xv_{1c}(l_{1,b-1})}_{P \Id_{N_2}} + {\Hm_{22} \uv_2(w_{2b})} + \underbrace{\Hm_{21} \uv_1(w_{1b})}_{\etav_{2b} \sim P^{A_1} \Id_{N'_2}} . 
\end{align}%
It is readily shown that $\etav_{1b}$ and $\etav_{2b}$ can be compressed
\emph{separately} up to the noise level with two independent source
codebooks of size $P^{N'_1 A_2}$ and $P^{N'_2 A_1}$, into indices
$l_{2,b}$ and $l_{1,b}$, respectively. For convenience, let us define 
\begin{align}
  d_{\eta_1} \defeq {N'_1 A_2}, \ 
  d_{\eta_2} \defeq {N'_2 A_1}, \ 
  \text{and } d_{\eta} \defeq d_{\eta_1} + d_{\eta_2}. 
  \label{eq:deta-ic}
\end{align}%

\subsection*{\underline{Backward decoding}}

The decoding starts at block $B$, since
$w_{1B}$ and $w_{2B}$ are both known, the private signals can be removed
from the received signals $\yv_1[B]$ and $\yv_2[B]$. The common messages
$l_{1,B-1}$ and $l_{2,B-1}$ can be decoded at both receivers if 
\begin{align}
  d_{\eta_k} &\le \min\left\{ M_j, N_1, N_2 \right\}, \\
  d_{\eta_1} + d_{\eta_2} &\le \min\left\{ N_1, N_2 \right\}, 
\end{align}%
i.e., the common rate pair should lie within the intersection of MAC
regions at both receivers for the common messages.  At block $b$,
assuming both $l_{1,b}$ and $l_{2,b}$ are known perfectly from the
decoding of block~$b+1$, $\etav_{1b}$ and $\etav_{2b}$ can be
reconstructed up to the noise level, for $b=B-1,\ldots,2$. The following MIMO system can be
obtained at Receiver~$k$
\begin{align} \label{eq:IC-MAC}
  \begin{bmatrix} \yv_k[b] - \etav_{kb} \\ \etav_{jb}
  \end{bmatrix} &= 
  \begin{bmatrix} \Hm_{kk} \\ 0 \end{bmatrix} \xv_{kc}(l_{k,b-1}) +
    \begin{bmatrix} \Hm_{kj} \\ 0 \end{bmatrix} \xv_{jc}(l_{j,b-1}) +
      \begin{bmatrix} \Hm_{kk} \\ \Hm_{jk} \end{bmatrix}
    \uv_k(w_{kb})
\end{align}%
for $j\ne k \in \left\{ 1, 2 \right\}$. 
Note that this system corresponds to a multiple-access channel from which the three
independent messages $l_{1,b-1}$, $l_{2,b-1}$, and $w_{kb}$ are to be decoded. 
It will be shown in the Appendix~\ref{app:MAC} that the three messages  
can be correctly decoded if the DoF
quadruple~$(d_{\eta_1}, d_{\eta_2}, d_{1b}, d_{2b})$ lies within the
following region
\begin{align}
   d_{kb} &\le N'_j A_k + (M_k-N'_j-\xi_k) A'_k + \xi_k A''_k \label{eq:MAC-dkb-ic}\\
   d_{\eta_k} &\le \min\left\{ M_j, N_1, N_2 \right\} \label{eq:MAC-detak-ic}\\
   d_{\eta_1} + d_{\eta_2} &\le \min\left\{ N_1, N_2 \right\} \label{eq:MAC-detasum-ic}\\
   d_{\eta_k} + d_{kb} &\le N'_k + \min\left\{ M_k - N'_j, N_k - N'_k
   \right\} A'_k + N'_j A_k \label{eq:MAC-detasum-ic2}\\
   d_{\eta_j} + d_{kb} &\le \min\left\{ M_k, N_k \right\} + N'_j A_k \label{eq:MAC-detasum-ic3} \\
   d_{\eta_1} + d_{\eta_2} + d_{kb} &\le N_k + N'_j A_k. \label{eq:MAC-allsum} 
\end{align}%
Now, let us fix
\begin{align}
   d_{kb} &\defeq N'_j A_k + (M_k-N'_j-\xi_k) A'_k + \xi_k A''_k  \label{eq:MAC-dkb-ic2}\\
   d_{\eta_j} &\defeq N'_j A_k \label{eq:MAC-detak-ic2}
\end{align}%
from which we can reduce the region defined by
\eqref{eq:MAC-dkb-ic}-\eqref{eq:MAC-allsum}. First, we remove
\eqref{eq:MAC-dkb-ic} that is implied by \eqref{eq:MAC-dkb-ic2}. Second, 
\eqref{eq:MAC-detak-ic} is not active as it is implied by
\eqref{eq:MAC-detak-ic2} and \eqref{eq:MAC-detasum-ic}. Third,
\eqref{eq:MAC-detasum-ic} is implied by \eqref{eq:MAC-allsum} and
\eqref{eq:MAC-dkb-ic2}. Finally,  
from \eqref{eq:MAC-detak-ic2}, \eqref{eq:MAC-detasum-ic3} is equivalent
to $d_{kb} \le \min\{M_k, N_k\}$ that is implied by
\eqref{eq:MAC-dkb-ic2}. Therefore, by letting $B\to\infty$, we have the following counterpart of
Lemma~\ref{lemma:decodability-BC} for interference channel.

\begin{lemma}[decodability condition for IC]\label{lemma:decodability-IC}
  Let us define 
  \begin{align} 
    \AsetIC &\defeq  \left\{ (A_1, A'_1, A''_1, A_2, A'_2, A''_2) \left| \Pmatrix{
    A_k, A'_k, A''_k \in [0,1] \\
    A'_k - \alpha_j \le A_k \le A_k', \ A''_k\le A'_k, \\ 
    \xi_k A''_k \le N_k' (1-A_j), \quad k\ne j\in\left\{ 1,2
    \right\} 
    } \right. \right\}
     \label{eq:Aset-IC} \\
    \DsetIC &\defeq \bigl\{ (d_1, d_2) \ \vert \ d_k \in [0,\nk], \quad
    \forall\,k \in \left\{ 1,2 \right\}  \bigr\}\\
    \shortintertext{and}
    \fAd \ : \ \AsetIC &\to \DsetIC \\
    (A_k, A'_k, A''_k) &\mapsto d_k \defeq N'_j A_k + (M_k-N'_j-\xi_k) A'_k + \xi_k
    A''_k, \quad \forall k\ne j \in\{1,2\} \label{eq:d1d2-ic}
  \end{align}%
  where
  \begin{align}
    \xi_k \defeq  \left\{ \Pmatrix{(M_k - N'_j)^+ - (N_k - N'_k)^+, & \text{if $C_k$ holds} \\ 0, & \text{otherwise.}} \right. \label{eq:xi}
  \end{align}
  Then $(d_1,d_2) = \fAd(\Avec)$, for some $\Am \in \AsetIC$, is achievable
  with the proposed scheme, if  
\begin{align}
  d_{\eta_1} + d_{1} &\le N_1, \label{eq:interp3} \\ 
  d_{\eta_2} + d_{2} &\le N_2. \label{eq:interp4}  
\end{align}%
where we recall $d_{\eta_1} \defeq N_1' A_2$ and $d_{\eta_2} \defeq N_2' A_1$. 
\end{lemma}
\begin{proof}
  It has been shown that with \eqref{eq:MAC-detak-ic2} and \eqref{eq:d1d2-ic},
  only \eqref{eq:MAC-detasum-ic2} and \eqref{eq:MAC-allsum} are active.
  With $\xi_k$ defined in \eqref{eq:xi}, we can verify that $M_k - N'_j
  - \xi_k = \min\left\{ M_k - N'_j, N_k - N'_k \right\}$. Thus,
  from \eqref{eq:MAC-detak-ic2}, \eqref{eq:d1d2-ic},  \eqref{eq:xi}, and
  the last constraint in \eqref{eq:Aset-IC}, it follows that \eqref{eq:MAC-detasum-ic2} always holds. Finally, the only
  constraint that remains is \eqref{eq:MAC-allsum} that can be
  equivalently written as \eqref{eq:interp3} and \eqref{eq:interp4}. 
\end{proof}

\begin{definition}[achievable region for IC]
  Let us define
  \begin{align} \label{eq:IAIC}
    \IAIC &\defeq  \left\{ (A_1, A'_1, A''_1, A_2, A'_2, A''_2) \in
    \AsetIC \left| \Pmatrix{
    (d_1, d_2) = \fAd(A_1, A'_1, A''_1, A_2, A'_2, A''_2), \\
    \frac{d_k}{N'_k} \le \frac{N_k}{N'_k} - A_j, \quad k\ne j\in\left\{ 1, 2 \right\}} \right. \right\}
  \end{align}%
and the achievable DoF region of the proposed scheme
  \begin{align} \label{eq:IdIC}
    \IdIC &\defeq \fAd(\IAIC) \defeq 
    \left\{ (d_1, d_2) \left| \Pmatrix{
    (d_1, d_2) = \fAd(A_1, A'_1, A''_1, A_2, A'_2, A''_2), \\
    (A_1, A'_1, A''_1, A_2, A'_2, A''_2) \in \AsetIC, \\
    \frac{d_k}{N'_k} \le \frac{N_k}{N'_k} - A_j, \quad k\ne j\in\left\{ 1, 2 \right\}} \right. \right\}.
  \end{align}%
\end{definition}

\subsection*{\underline{Achievability analysis}}

The analysis is similar to the BC case, i.e., it is sufficient to find a
function $\fdA \,:\, \OdIC \to \AsetIC$ where $\OdIC$ denotes the outer
bound region defined by \eqref{eq:outer-bound-ic}, such that
\begin{align}
  (d_1, d_2) &= \fAd(\fdA(d_1, d_2)),\quad \text{and} \label{eq:inverse-IC} \\
  \fdA(d_1, d_2) &\in \IAIC. \label{eq:cond-dec-IC}
\end{align}

Now, we define formally the power allocation function. 
\begin{definition}[power allocation for IC] \label{def:Ak-dk-ic}
Let us define $\gamma_k$, $k\ne j \in \{1,2\}$, as
\begin{align}
  \gamma_k &\defeq \min\left\{1, \frac{M_j - d_j}{\xi_k}  \right\}. 
\end{align}%
Then, we define $\fdA\ : \ \OdIC \to \AsetIC$:
\begin{align}
  (d_1, d_2) \mapsto (A_1, A'_1, A''_1) \defeq f_1(d_1, d_2),\ (A_2,
  A'_2, A''_2) \defeq f_2(d_1, d_2)
\end{align}%
where $f_k$, $k\ne j \in \{1,2\}$, such that \eqref{eq:d1d2-ic} is satisfied,
and that 
\begin{itemize}
  \item when $M_k=N'_j$: $A''=A'_k=A_k = \frac{d_k}{M_k}$;
  \item when $M_k>N'_j$, $d_k < (M_k - N'_j) \gamma_k + N'_j
    (\gamma_k-\alpha_j)^+$:
    \begin{align}
      A_k &= (A'_k - \alpha_j)^+, \quad A'_k = A''_k < \gamma_k; \label{eq:Ak-ic-Ck}
    \end{align}%
  \item when $M_k>N'_j$, $d_k \ge (M_k - N'_j) \gamma_k + N'_j
    (\gamma_k-\alpha_j)^+$, and $\gamma_k < 1$:
    \begin{align}
      A_k &= (A'_k - \alpha_j)^+, \quad A'_k > A''_k = \gamma_k; \label{eq:Ak-ic-Ck-1}
    \end{align}%
  \item when $M_k>N'_j$, $d_k \ge (M_k - N'_j) \gamma_k + N'_j
    (\gamma_k-\alpha_j)^+$, and $\gamma_k = 1$:
    \begin{align}
      A'_k &= A''_k = 1. \label{eq:Ak-ic-Ck-2}
    \end{align}%
\end{itemize}
\end{definition}
First, one can verify, with some basic manipulations that, $\fdA(\OdIC) \subseteq \AsetIC$. Second, \eqref{eq:inverse-IC} is
satisfied by construction.  Finally, it remains to show
that \eqref{eq:cond-dec-IC} holds as well, i.e., the decodability
condition in \eqref{eq:IAIC} is satisfied. 
Since the region $\OdIC$ depends on whether the condition $C_k$ holds,
we prove the achievability accordingly. 
 
\subsubsection{Neither $C_1$ nor $C_2$ holds~($\xi_1 = \xi_2 = 0$)} 

For any $(d_1, d_2) \in \OdIC$, 
we can define $(A_1,A'_1,A''_1,A_2,A_2',A''_2) \defeq \fdA(d_1,d_2)$
which implies, in this case, 
 \begin{align}
   d_j &= N_k' A_j + (M_j-N_k')A'_j,\quad  j \ne k\in\{1,2\}. \label{eq:d1d2-ic2}
 \end{align}%
Applying this equality on the constraints in the outer bound $\OdIC$ in
\eqref{eq:outer-bound-ic}, we have 
\begin{align}
  \frac{d_k}{N'_k} &\le \frac{\min\left\{ \max\left\{ M_1,N_2 \right\},
  \max\left\{ M_2,N_1 \right\}\right\} - (M_j-N'_k)A'_j}{N'_k} - A_j, \label{eq:tmp777}\\
  \frac{d_k}{N'_k} &\le \frac{\min\{M_k, N_k\}}{N'_k} -
  \left[\frac{\min\{M_k, N_k\} - N_k}{N'_k} + \frac{(M_j-N'_k)(A'_j-\alpha_k) +
   N'_k A_j }{M_j}  \right]^+, \label{eq:tmp778}
\end{align}%
for $k\ne j\in\left\{ 1, 2 \right\}$, where the first one is from the
sum rate constraint \eqref{eq:ic-outer-bound-3} whereas the second one
is from the rest of the constraints in \eqref{eq:outer-bound-ic}. 
The final step is to show that either of \eqref{eq:tmp777} and
\eqref{eq:tmp778} implies the last constraint in \eqref{eq:IAIC}. 
\begin{itemize}
  \item When $M_j=N'_k$, \eqref{eq:tmp778} implies the last
    constraint in \eqref{eq:IAIC} because $\frac{\nk - N_k}{N_k} \le 0$; 
  \item When $M_j>N'_k$ and $d_j \ge M_j - N'_k \alpha_k$, we have $A'_j=1$
    according to the map $\fdA$ defined in
    Definition~\ref{def:Ak-dk-ic},
    since $\gamma_j = 1$. Hence, 
    the right hand side~(RHS) of \eqref{eq:tmp777} is not greater than
    $\frac{N_k}{N'_k} - A_j$, which implies the last
    constraint in \eqref{eq:IAIC}; 
  \item When $M_j > N'_k$ and $d_j < M_j - N'_k \alpha_k$, we have $A_j =
    (A'_j - \alpha_k)^+$ according to Definition~\ref{def:Ak-dk-ic} with $\gamma_j = 1$. Since
    $\frac{\nk - N_k}{N_k} \le 0$, we can show
    that 
    \begin{align} \label{eq:const-ic-1}
      \left[\frac{\nk - N_k}{N_k} + \frac{(M_j-N'_k)(A'_j-\alpha_k) +
      N'_k A_j }{M_j}  \right]^+ &\ge \frac{\nk
      - N_k}{N_k} + A_j
    \end{align}%
    with which \eqref{eq:tmp778} implies the last constraint in
    \eqref{eq:IAIC}. 
\end{itemize}

\subsubsection{ $C_k$ holds~($\xi_k>0$, $\xi_j=0$)}
In this case, it is readily shown, from \eqref{eq:d1d2-ic} and
\eqref{eq:xi}, that
\begin{align}
  d_k &= N_j A_k + (N_k - M_j)A'_k + \xi_k A''_k, \label{eq:dk-Ck}\\
  d_j &= M_j A_j. \label{eq:dj-Ck}
\end{align}
Applying the map $d_j = M_j A_j$ on \eqref{eq:ic-outer-bound-3} results in
\begin{align}
  \frac{d_k}{N'_k} &\le \frac{\nk}{N'_k} -
  A_j\label{eq:tmp888}
\end{align}%
that always implies $\frac{d_k}{N'_k} \le \frac{N_k}{N'_k} - A_j$. 
Due to the asymmetry, we also need to prove that $\frac{d_j}{N_j} \le
1 - A_k$. 
Therefore, the final step is to show that it can be implied by at least one of the constraints in \eqref{eq:outer-bound-ic}, together
with \eqref{eq:dk-Ck} and \eqref{eq:dj-Ck}.
\begin{itemize}
  \item When $d_k < (M_k-N_j) \gamma_k + N_j(\gamma_k - \alpha_j)^+$, we
    have $A'_k=A''_k < \gamma_k$ according to \eqref{eq:Ak-ic-Ck}. Therefore,
    $d_k = N_j A_k + (M_k-N_j)A'_k$, plugging
    which into \eqref{eq:ic-outer-bound-5}, we obtain 
    \begin{align} 
      \frac{d_j}{N_j} &\le \frac{\nj}{N_j} - \left[ \frac{\nj-N_j}{N_j}
      + \frac{(M_k-N_j)(A'_k-\alpha_j) + N_jA_k}{M_k}\right]^+
      \label{eq:tmp889} \\
      &\le \frac{\nj}{N_j} - \left[\frac{\nj-N_j}{N_j} + A_k \right] 
    \end{align}%
    where the $[\cdot]^+$ in \eqref{eq:tmp889} is from the single user
    bound \eqref{eq:ic-outer-bound-2}; the last inequality is due to 
    $A_k=(A'_k-\alpha_j)^+$ and $\frac{\nj-N_j}{N_j}\le0$.
  \item When $d_k \ge (M_k-N_j) \gamma_k+ N_j(\gamma_k - \alpha_j)^+$, we have
    $A'_k \ge A''_k=\gamma_k$ according to \eqref{eq:Ak-ic-Ck-1} and \eqref{eq:Ak-ic-Ck-2}. 
    \begin{itemize}
      \item If $\gamma_k < 1$, then $A''_k = \gamma_k = \frac{M_j -
        d_j}{\xi_k}$ and $d_k = (N_k-M_j)A'_k + M_j - d_j + N_j A_k$.
        Plugging the latter into \eqref{eq:ic-outer-bound-6}, we obtain 
        \begin{align}
          \frac{d_j}{N_j} &\le \frac{\nj}{N_j} - \left[ \frac{\nj-N_j}{N_j} + \frac{(N_k-M_j)(A'_k-\alpha_j) + N_j A_k}{N_k+N_j-M_j}  \right]^+ \label{eq:tmp890} \\
          &\le \frac{\nj}{N_j} - \left[ \frac{\nj-N_j}{N_j} + A_k \right]
        \end{align}%
        where the $[\cdot]^+$ in \eqref{eq:tmp890} is from the single user
        bound \eqref{eq:ic-outer-bound-2}; the last inequality is due to 
        $A_k=(A'_k-\alpha_j)^+$ and $\frac{\nj-N_j}{N_j}\le0$.
      \item If $\gamma_k = 1$, then $A'_k = A''_k = 1$ and $d_k = M_k - N_j +
        N_j A_k$. Plugging the latter into \eqref{eq:ic-outer-bound-3}, we
        obtain 
        \begin{align}
          \frac{d_j}{N_j} &\le \frac{\nk - M_k + N_j - N_j A_k}{N_j} \\
          &\le {1 -  A_k}. \label{eq:tmp1007}
        \end{align}%
    \end{itemize}
\end{itemize}
Thus, the last constraint in \eqref{eq:IAIC} is shown in all cases. By now, we have proved the achievability through the existence of a
proper power allocation function such that \eqref{eq:inverse-IC} and
\eqref{eq:cond-dec-IC} are satisfied for every pair $(d_1, d_2)$ in the
outer bound.

\section{Converse}
To obtain the outer bounds, the following ingredients are essential:
\begin{itemize}
  \item \emph{Genie-aided} bounding techniques by providing side information of one receiver to the other one~\cite{Vaze:2011BC,Vaze:2012IC};
  \item \emph{Extremal inequality} to bound the weighted difference of conditional differential entropies~\cite{Liu:2007,Extremal};
  \item \emph{Ergodic capacity upper and lower bounds} for MIMO channels with channel uncertainty.
\end{itemize}

In the following, we first present the proof of outer bound \eqref{eq:bc-outer-bound-4} for MIMO BC and \eqref{eq:ic-outer-bound-4} for MIMO IC, referred to in this section as $L_4$. It should be noticed that both bounds share the same structure. Then, we give the proof of bound \eqref{eq:ic-outer-bound-6} for the MIMO IC case, referred to in this section as $L_6$, when the condition $C_1$ holds.

\subsection{Proof of Bound $L_4$}
We first provide the outer bounds by employing the genie-aided
techniques for BC and IC, respectively, reaching a similar {formulation}
of the rate bounds. With the help of extremal inequalities, the weighted
sum rates are further bounded. Finally, the bounds in terms of
$(\alpha_1, \alpha_2)$ are obtained by deriving novel ergodic capacity bounds for MIMO channels with channel uncertainty. 

To obtain the outer bounds, we adopt a genie-aided upper bounding technique reminisced in \cite{Vaze:2011BC,Vaze:2012IC}, by providing Receiver~2 the side information of the Receiver~1's message $W_1$ as well as received signal $\Ym_1^n$. 
For notational brevity, we define {a virtual received signal as}
\begin{align}
\bar{\yv}_i(t) \defeq \left\{
\Pmatrix{
\Hm_{i}(t) \xv(t) + \zv_i(t) & \text{for BC}\\
\Hm_{i2}(t) \xv_2(t) + \zv_i(t) & \text{for IC}
} \right. \label{eq:def-eq}
\end{align}
and we also define $\Xm^n \defeq \{\xv(t)\}_{t=1}^n$, $\Xm_i^n \defeq \{\xv_i(t)\}_{t=1}^n$, $\Ym_i^k \defeq \{\yv_i(t)\}_{t=1}^k$, and $\bar{\Ym}_i^k \defeq \{\bar{\yv}_i(t)\}_{t=1}^k$. Denote also $n\epsilon_n \defeq 1+nR P_{e}^{(n)}$ where $\epsilon_n$ tends to zero as $n \to \infty$ by the assumption that $\lim_{n \to \infty}P_{e}^{(n)}=0$.

\subsubsection{Broadcast Channel}
The genie-aided model is a degraded BC $\Xm^n \to \left(\Ym_1^n, \Ym_2^n\right) \to \Ym_1^n$, and therefore we bound the achievable rates by applying Fano's inequality as
\begin{align}
n(R_1-\epsilon_n) &\le I(W_1;\Ym_1^n|\Hc^n,\hat{\Hc}^n)\\
&= \sum_{t=1}^n I(W_1; \yv_1(t)|\Hc^n,\hat{\Hc}^n,\Ym_1^{t-1}) \\
&= \sum_{t=1}^n \left(h(\yv_1(t)|\Hc^n,\hat{\Hc}^n,\Ym_1^{t-1})-h(\yv_1(t)|\Hc^n,\hat{\Hc}^n,\Ym_1^{t-1},W_1)\right) \\
&\le \sum_{t=1}^n \left(h({\yv}_1(t)|\Hc(t))-h({\yv}_1(t)|\Uc(t),\Hc(t)) \right)  \label{eq:BC-bound-1}\\
&\le n N_1' \log P - \sum_{t=1}^n h( \bar{\yv}_1(t)|\Uc(t),\Hc(t)) + n \cdot O(1)  \label{eq:BC-bound-2}
\\
n(R_2-\epsilon_n)
&\le I(W_2;\Ym_1^n,\Ym_2^n,W_1|\Hc^n,\hat{\Hc}^n) \\
&= I(W_2;\Ym_1^n,\Ym_2^n|W_1,\Hc^n,\hat{\Hc}^n)  \label{eq:BC-bound-3}\\
&= \sum_{t=1}^n I(W_2;\yv_1(t),\yv_2(t)|\Hc^n,\hat{\Hc}^n,\Ym_1^{t-1},\Ym_2^{t-1},W_1)   \\
&\le \sum_{t=1}^n  I(\xv(t);\yv_1(t),\yv_2(t)|\Hc^n,\hat{\Hc}^n,\Ym_1^{t-1},\Ym_2^{t-1},W_1)  \label{eq:BC-bound-4}\\
&= \sum_{t=1}^n \left(h(\yv_1(t),\yv_2(t)|\Hc^n,\hat{\Hc}^n,{\Ym}_1^{t-1},\Ym_2^{t-1},W_1) \right. \\
& \left. ~~~~~~~~~~~~~~~-h(\yv_1(t),\yv_2(t)|\xv(t),\Hc^n,\hat{\Hc}^n,\Ym_1^{t-1},\Ym_2^{t-1} ,W_1) \right)  \\
&\le\sum_{t=1}^n h(\yv_1(t),\yv_2(t)|\Hc^n,\hat{\Hc}^n,\Ym_1^{t-1},\Ym_2^{t-1},W_1)  \label{eq:BC-bound-5}\\
&=\sum_{t=1}^n h(\bar{\yv}_1(t),\bar{\yv}_2(t)|\Uc(t),\Hc(t))
\end{align}
where $N_1' \defeq \min\{M,N_1\}$ and $\Uc(t) \defeq \left\{
\bar{\Ym}_1^{t-1},\bar{\Ym}_2^{t-1},\Hc^{t-1} ,\hat{\Hc}^{t}, W_1
\right\}$  for BC; (\ref{eq:BC-bound-1}) is from (\ref{eq:def-eq}) and
because (a) conditioning reduces differential entropy, and (b)
$\{\bar{\yv}_1(t),\bar{\yv}_2(t)\}$ are independent of $\Hc_{t+1}^n$ and
$\hat{\Hc}_{t+1}^n$, given the past states and channel outputs;
(\ref{eq:BC-bound-2}) follows the fact that the rate of the
point-to-point $M \times N_1$ MIMO channel (i.e., between the Transmitter an Receiver~1) is
bounded by $ \min\{M,N_1\}\log P + O(1)$; (\ref{eq:BC-bound-3}) is due
to the independence between $W_1$ and $W_2$; (\ref{eq:BC-bound-4})
follows date processing inequality; and (\ref{eq:BC-bound-5}) is
obtained by noticing (a) translation does not change differential
entropy, (b) Gaussian noise terms are independent from instant to
instant, and are also independent of the channel matrices and the
transmitted signals, and (c) the differential entropy of Gaussian noise
with normalized variance is nonnegative.
\subsubsection{Interference Channel}
Given the message and corresponding channel states, part of the received signal is deterministic and therefore removable without mutual information loss. Hence, similarly to the BC case, we formulate a degraded channel model, i.e., $\Xm_2^n \to \left(\bar{\Ym}_1^n, \bar{\Ym}_2^n\right) \to \bar{\Ym}_1^n$. By applying Fano's inequality, the achievable rate of Receiver~1 and Receiver~2 can be bounded as
\begin{align}
n(R_1-\epsilon_n) &\le I(W_1;\Ym_1^{n}|\Hc^n,\hat{\Hc}^n) \\
&= I(W_1,W_2;\Ym_1^{n}|\Hc^n,\hat{\Hc}^n) - I(W_2;\Ym_1^{n}|W_1,\Hc^n,\hat{\Hc}^n)   \\
&\le n \tilde{N}_1 \log P - I(W_2;\Ym_1^{n}|W_1,\Hc^n,\hat{\Hc}^n) + n \cdot O(1)  \label{eq:IC-bound-1}\\
&= n \tilde{N}_1 \log P - h(\Ym_1^{n}|W_1,\Hc^n,\hat{\Hc}^n) + h(\Ym_1^{n}|W_1,W_2,\Hc^n,\hat{\Hc}^n) + n \cdot O(1) \\
&= n \tilde{N}_1 \log P - h(\Ym_1^{n}|W_1,\Hc^n,\hat{\Hc}^n) + n \cdot O(1)  \label{eq:IC-bound-2}\\
&= n \tilde{N}_1 \log P - h(\bar{\Ym}_1^{n}|\Hc^n,\hat{\Hc}^n) + n \cdot O(1)   \label{eq:IC-bound-3}\\
&\le n \tilde{N}_1 \log P - \sum_{t=1}^n h( \bar{\yv}_1(t)|\Hc^n,\hat{\Hc}^n,\bar{\Ym}_1^{t-1},\bar{\Ym}_2^{t-1}) + n \cdot O(1)  \label{eq:IC-bound-3.0}\\
&= n \tilde{N}_1 \log P - \sum_{t=1}^n h( \bar{\yv}_1(t)|\Uc(t),\Hc(t)) + n \cdot O(1)
\\
n(R_2-\epsilon_n)
&\le I(W_2;\Ym_1^{n},\Ym_2^{n},W_1|\Hc^n,\hat{\Hc}^n)  \\
&= I(W_2;\Ym_1^{n},\Ym_2^{n}|W_1,\Hc^n,\hat{\Hc}^n)  \\
&= I(W_2;\bar{\Ym}_1^{n},\bar{\Ym}_2^{n}|\Hc^n,\hat{\Hc}^n)  \label{eq:IC-bound-4}\\
&= \sum_{t=1}^n I(W_2;\bar{\yv}_1(t),\bar{\yv}_2(t)|\Hc^n,\hat{\Hc}^n,\bar{\Ym}_1^{t-1},\bar{\Ym}_2^{t-1})   \\
&\le \sum_{t=1}^n  I(\xv_2(t);\bar{\yv}_1(t),\bar{\yv}_2(t)|\Hc^n,\hat{\Hc}^n,\bar{\Ym}_1^{t-1},\bar{\Ym}_2^{t-1})  \\
&= \sum_{t=1}^n \left(h(\bar{\yv}_1(t),\bar{\yv}_2(t)|\Hc^n,\hat{\Hc}^n,\bar{\Ym}_1^{t-1},\bar{\Ym}_2^{t-1}) \right. \\
& \quad \quad \quad \quad\quad\quad\quad\quad - \left. h(\bar{\yv}_1(t),\bar{\yv}_2(t)|\xv_2(t),\Hc^n,\hat{\Hc}^n,\bar{\Ym}_1^{t-1},\bar{\Ym}_2^{t-1} ) \right)  \\
&\le \sum_{t=1}^n h(\bar{\yv}_1(t),\bar{\yv}_2(t)|\Hc^n,\hat{\Hc}^n,\bar{\Ym}_1^{t-1},\bar{\Ym}_2^{t-1})  \label{eq:IC-bound-5}\\
&=\sum_{t=1}^n h(\bar{\yv}_1(t),\bar{\yv}_2(t)|\Uc(t),\Hc(t))  \label{eq:IC-bound-6}
\end{align}
where we define $\Uc(t) \defeq \left\{
\bar{\Ym}_1^{t-1},\bar{\Ym}_2^{t-1},\Hc^{t-1} ,\hat{\Hc}^{t}\right\}$
for IC and $\tilde{N}_1 \defeq \min\{M_1+M_2,N_1\}$;
(\ref{eq:IC-bound-1}) follows the fact that the mutual information at
hand is upper bounded by the rate of the $(M_1+M_2) \times N_1$
point-to-point MIMO channel created by letting the two transmitters
cooperate perfectly, given by $\min\{M_1+M_2,N_1\} \log P + O(1)$;
(\ref{eq:IC-bound-2}) is due to the fact that (a) transmitted signal
$\Xm_i^n$  is a deterministic function of messages $W_i$, $\Hc^n$, and
$\hat{\Hc}^{n-1}$ as specified in \eqref{eq:enc-fun} for $i=1,2$, (b)
translation does change differential entropy, and (c) the
differential entropy of Gaussian noise with normalized variance is
nonnegative and finite; (\ref{eq:IC-bound-3}) and
(\ref{eq:IC-bound-4}) are obtained because translation preserves
differential entropy; (\ref{eq:IC-bound-3.0}) is because conditioning
reduces differential entropy; (\ref{eq:IC-bound-5}) is because (a)
translation does not change differential entropy, (b) Gaussian noise
terms are independent from instant to instant, and are also independent
of the channel matrices and the transmitted signals, and (c) the
differential entropy of Gaussian noise with normalized variance is
nonnegative and finite; and (\ref{eq:IC-bound-6}) is obtained due to the
independence $\{\bar{\yv}_1(t),\bar{\yv}_2(t)\}$ of $\Hc_{t+1}^n$ and
$\hat{\Hc}_{t+1}^n$, given the past state and channel outputs.

It is worth noting that BC and IC share the common structure of the achievable rate bounds, and therefore can be further bounded in a similar way. To avoid redundancy, we give the proof for IC, which can be straightforwardly extended to BC.

Define
\begin{align}
  \Sm(t) \defeq \Bmatrix{\Hm_{12}(t) \\ \Hm_{22}(t)} \quad \quad \hat{\Sm}(t) \defeq \Bmatrix{\hat{\Hm}_{12}(t) \\ \hat{\Hm}_{22}(t)}  \quad \quad
  \Km(t) \eqdef \E\{\xv_2(t)\xv_2^{\H}(t) \cond \Uc(t)\}.
\end{align} 
Let $p=\min\{M_2,N_1+N_2\}$ and $q=\min\{M_2,N_1\}$. By following the
footsteps in \cite{Yang:2013}, we have
{\small
\begin{align}
  \MoveEqLeft \frac{1}{p}h(\bar{\yv}_1(t),\bar{\yv}_2(t)|\Uc(t),\Hc(t)) - \frac{1}{q} h(\bar{\yv}_1(t)|\Uc(t),\Hc(t))\\
&\le \E_{\hat{\Sm}(t)}  \max_{\substack{\Km \succeq 0,\\ \trace(\Km) \le P}} \E_{\Sm(t)|\hat{\Sm}(t)} \left(\frac{1}{p}\log \det (\mathbf I + \Sm(t) \Km(t) \Sm^{\H}(t) ) - \frac{1}{q} \log \det (\Id+\Hm_{12}(t) \Km(t) \Hm_{12}^\H(t))\right) \label{eq:EI-3} \\
   &\le -\frac{\min\{M_2,N_1+N_2\}-\min\{M_2,N_1\}}{\min\{M_2,N_1+N_2\}} \log \sigma_1^2 + O(1)
\end{align}%
}%
where (\ref{eq:EI-3}) is obtained by applying extremal inequality
\cite{Liu:2007,Extremal} for degraded outputs;
 the last inequality is obtained from the following lemma:
\begin{lemma} \label{lemma:gcase}
  For two random matrices $\Sm = \hat{\Sm}+\tilde{\Sm} \in \CC^{L \times
  M}$ and $\Hm = \hat{\Hm}+\tilde{\Hm} \in \CC^{N \times M}$ with $L \ge
  N$, $\tilde{\Sm}$, $\tilde{\Hm}$ being
  respectively independent of $\hat{\Sm}$, $\hat{\Hm}$, 
  the entries of $\tilde{\Hm}$ being i.i.d. $\CN(0,\sigma^2)$. Then, given any $\Km
  \succeq 0$ with eigenvalues $\lambda_1 \ge \cdots \ge \lambda_{M} \ge
  0$, we have 
  \begin{multline}
    \frac{1}{\min\{M,L\}} \E_{\tilde{\Sm}} \log \det (\Id+\Sm \Km
    \Sm^\H) - \frac{1}{\min\{M,N\}} \E_{\tilde{\Hm}} \log \det (\Id+\Hm
    \Km \Hm^\H)  \\
    \le - \frac{\min\{M,L\}-\min\{M,N\}}{\min\{M,L\}} \log (\sigma^2) +
    \Es + \Eh \label{eq:lemma3}
  \end{multline}
  as $\sigma^2$ goes to 0.
\end{lemma}

\begin{proof}
  See Appendix B. 
\end{proof}

{\bf Remark}:
\begin{itemize}
\item This lemma can be regarded as the weighted difference of the ergodic capacity for two MIMO channels with uncertainty, where $\tilde{\Sm}$ and $\tilde{\Hm}$ are channel uncertainties. It can also be interpreted as the ergodic capacity difference of two Ricean MIMO channels with line-of-sight components $\hat{\Sm}$, $\hat{\Hm}$, and fading components $\tilde{\Sm}$, $\tilde{\Sm}$. 
\item This lemma also shows the change of the ergodic capacity per dimension as the dimensionality decreases. In other words, as the channel dimension decreases, the difference of the ergodic capacity per dimension is bounded by the dimension difference and the channel uncertainty.
\end{itemize} 

According to the Markov chain $\Xm_2^n \to \left(\bar{\Ym}_1^n, \bar{\Ym}_2^n\right) \to \bar{\Ym}_1^n$, we upper-bound the weighted sum rate as
\begin{align}
  \MoveEqLeft n \left( \frac{R_1}{\min\{M_2,N_1\}} + \frac{R_2}{\min\{M_2,N_1+N_2\}} - \epsilon_n \right) \\
  &\le  n \cdot \frac{\min\{M_1+M_2,N_1\}}{\min\{M_2,N_1\}} \log P +  \sum_{t=1}^n \left( \frac{1}{\min\{M_2,N_1+N_2\}}h(\bar{\yv}_1(t),\bar{\yv}_2(t)|\Uc(t),\Hc(t)) \right. \\ &\left.~~~~~~~~~~~~~~~~~~~~~~~~~~~~~~~~~~~~~~~~ - \frac{1}{\min\{M_2,N_1\}} h( \bar{\yv}_1(t)|\Uc(t),\Hc(t))\right) + n \cdot O(1) \\
  &\le n \frac{\min\{M_1+M_2,N_1\}}{\min\{M_2,N_1\}} \log P + n \frac{\min\{M_2,N_1+N_2\}-\min\{M_2,N_1\}}{\min\{M_2,N_1+N_2\}}\alpha_1 \log P  + n \cdot O(1) 
\end{align}
and another outer bound can be similarly obtained by exchanging the roles of Receiver~1 and Receiver~2.
Accordingly, the corresponding outer bound $L_4$ of the DoF region is obtained by the definition.

\subsection{Proof of Bound $L_6$}
This bound is active when $C_1$ holds, i.e., $M_1\ge N_2$, $N_1>M_2$,
and $M_1+M_2>N_1+N_2$. The proof follows the same lines of thought in
\cite{Vaze:2012IC}. Since $N_1 > M_2$, we formulate a virtual received signal
\begin{align}
\tilde{\yv}_1(t) \defeq \Um \yv_1(t) = \Um \Hm_{11}(t) \xv_1(t) + \Um
\Hm_{12}(t) \xv_2(t) + \Um \zv_1(t)
\end{align}
where $\Um \in \CC^{N_1 \times N_1}$ is any unitary matrix such that the
last $N_1-M_2$ rows of $\Um(t) \Hm_{12}(t)$ are with all zeros and that
is independent of the rest of random variables. Therefore, the last
$N_1-M_2$ outputs in $\tilde{\yv}_1(t)$ are interference free, i.e.,
$\tilde{\yv}_{1[M_2+1:N_1]}(t) \sim  \Hm_{1[M_2+1:N_1]1}(t) \xv_1(t) +
\zv_{1[M_2+1:N_1]}(t)$.  
For convenience, we also define
\begin{align}
\tilde{\yv}_2(t) &\defeq \Hm_{21}(t) \xv_1(t) + \zv_2(t). 
\end{align}

Starting with Fano's inequality, the achievable rate can be bounded as
\begin{align}
n(R_1-\epsilon_n) &\le I(W_1;\Ym_1^{n}|\Hc^n,\hat{\Hc}^n) \\
& = I(W_1;\tilde{\Ym}_1^{n}|\Hc^n,\hat{\Hc}^n) \label{eq:L7-1}\\
&=
I(W_1;\tilde{\Ym}_{1[1:M_2]}^{n}|\Hc^n,\hat{\Hc}^n,\tilde{\Ym}_{1[M_2+1:N_1]}^{n})
+ I(W_1;\tilde{\Ym}_{1[M_2+1:N_1]}^{n}|\Hc^n,\hat{\Hc}^n) \\
&\le n(M_2 - d_2) \log P +  I(W_1;\tilde{\Ym}_{1[M_2+1:N_1]}^{n}|\Hc^n,\hat{\Hc}^n) + O(1) \label{eq:L7-2}\\
&\le n(M_2 - d_2) \log P +  I(W_1;\tilde{\Ym}_{1[M_2+1:N_1]}^{n}, \tilde{\Ym}_2^n|\Hc^n,\hat{\Hc}^n) + O(1) \\
&=  n(M_2 - d_2) \log P + \sum_{t=1}^n I(W_1;
\tilde{\yv}_{1[M_2+1:N_1]}(t),
\tilde{\yv}_2(t)|\Hc^n,\hat{\Hc}^n,\tilde{\Ym}_{1[M_2+1:N_1]}^{t-1},
\tilde{\Ym}_2^{t-1}) + O(1) \\
&\le  n(M_2 - d_2) \log P + \sum_{t=1}^n
h(\tilde{\yv}_{1[M_2+1:N_1]}(t),
\tilde{\yv}_2(t)|\Hc^n,\hat{\Hc}^n,\tilde{\Ym}_{1[M_2+1:N_1]}^{t-1},
\tilde{\Ym}_2^{t-1}) + O(1) \label{eq:L7-3} \\
&=  n(M_2 - d_2) \log P + \sum_{t=1}^n h(\tilde{\yv}_{1[M_2+1:N_1]}(t),
\tilde{\yv}_2(t)|\Uc(t),\Hc(t)) + O(1) \label{eq:L7-30} \\
n(R_2-\epsilon_n) &\le I(W_2;\Ym_2^{n}|\Hc^n,\hat{\Hc}^n) \\
&= I(W_1,W_2;\Ym_2^{n}|\Hc^n,\hat{\Hc}^n) -  I(W_1;\Ym_2^{n}|W_2, \Hc^n,\hat{\Hc}^n) \\
&\le nN_2 \log P -  I(W_1;\tilde{\Ym}_2^{n}|\Hc^n,\hat{\Hc}^n) + O(1) \label{eq:L7-4}\\
&\le nN_2 \log P -  h(\tilde{\Ym}_2^{n}|\Hc^n,\hat{\Hc}^n) + h(\tilde{\Ym}_2^{n}|W_1, \Hc^n,\hat{\Hc}^n)+ O(1)\\
&= nN_2 \log P - h(\tilde{\Ym}_2^{n}|\Hc^n,\hat{\Hc}^n) + O(1) \label{eq:L7-5}\\
&\le nN_2 \log P - \sum_{t=1}^n h(\tilde{\yv}_2(t)|\Hc^n,\hat{\Hc}^n,
\tilde{\Ym}_{1[M_2+1:N_1]}^{t-1}, \tilde{\Ym}_2^{t-1}) + O(1) \label{eq:L7-6} \\
&= nN_2 \log P - \sum_{t=1}^n h(\tilde{\yv}_2(t)|\Uc(t),\Hc(t)) + O(1)
\end{align}
where $\tilde{\Ym}_i^k \defeq \{\tilde{\yv}_i(t)\}_{t=1}^k$, $i=1,2$, and $\Uc(t) \defeq
\{\Hc^{t-1},\hat{\Hc}^t,\tilde{\Ym}_{1[M_2+1:N_1]}^{t-1}, \tilde{\Ym}_2^{t-1}\}$; \eqref{eq:L7-1} holds due to the fact that unitary transformation does not change the mutual information; \eqref{eq:L7-2} comes from the Lemma 6 in \cite{Vaze:2012IC}, given by
\begin{align}
I(W_1;\tilde{\Ym}_{1[1:M_2]}^{n}|\Hc^n,\hat{\Hc}^n,\tilde{\Ym}_{1[M_2+1:N_1]}^{n}) \le n(M_2 - d_2) \log P + O(1)
\end{align}%
where a similar proof can be straightforwardly obtained;
 \eqref{eq:L7-3} holds because (a) $ \tilde{\yv}_{1[M_2+1:N_1]}(t)$ and
 $\tilde{\yv}_2(t)$ are deterministic functions of $W_1$, $\Hc^n$ and
 $\hat{\Hc}^n$, (b) translation does not change differential entropy,
 and (c) the differential entropy of Gaussian noise with normalized
 variance is non-negative;
 \eqref{eq:L7-4} follows that the mutual information at hand is upper
 bounded by the capacity of an $(M_1+M_2) \times N_2$ point-to-point MIMO
 channel, i.e., $N_2 \log P + O(1)$ since $M_1+M_2
 > N_2$ from condition $C_1$; \eqref{eq:L7-5} holds because
 $\tilde{\yv}_2(t)$ is a deterministic function of $W_1$, given channel
 states, and the differential entropy of the normalized Gaussian noise
 is finite; \eqref{eq:L7-6} is due to conditioning reduces the
 differential entropy; \eqref{eq:L7-30} and the last equality are due to
 that the received signals at instant $t$ are independent of
 $\Hc^{n}_{t+1}$ and $\hat{\Hc}^{n}_{t+1}$, given the past states and
 channel outputs.

Next, we define
\begin{align}
\Sm(t) \defeq \Bmatrix{\Hm_{1[M_2+1:N_1]1}(t) \\ \Hm_{21}(t)} \in
\CC^{(N_1+N_2-M_2) \times M_1}.
\end{align}
Similarly to the proof for bound $L_4$, we obtain the weighted
difference of two differential entropies by applying the extremal
inequality and Lemma~\ref{lemma:gcase}
\begin{align}
  \frac{1}{p}h(\tilde{\yv}_{1[M_2+1:N_1]}(t),\tilde{\yv}_2(t)|\Uc(t),\Hc(t)) - \frac{1}{q} h(\tilde{\yv}_2(t)|\Uc(t),\Hc(t))
&\le \frac{N_1-M_2}{N_1+N_2-M_2}  \log \sigma_2^2 + O(1)
\end{align} %
where we set $p=\min\{M_1,N_1+N_2-M_2\}=N_1+N_2-M_2$ and
$q=\min\{M_1,N_2\}=N_2$.

Finally, we have
\begin{align}
  \MoveEqLeft
  n \left( \frac{R_1}{N_1+N_2-M_2} + \frac{R_2}{N_2} -  \epsilon_n\right) \\
  &\le  n \left( 1+\frac{M_2-d_2}{N_1+N_2-M_2} \right) \log P +  \sum_{t=1}^n \left( \frac{1}{N_1+N_2-M_2}h(\tilde{\yv}_{1[M_2+1:N_1]}(t),\tilde{\yv}_2(t)|\Uc(t),\Hc(t)) \right. \\ &\left.~~~~~~~~~~~~~~~~~~~~~~~~~~~~~~~~~~~~~~~~ - \frac{1}{N_2} h(\tilde{\yv}_2(t)|\Uc(t),\Hc(t))\right) + n \cdot O(1) \\
  &\le n \left( 1+\frac{M_2-d_2}{N_1+N_2-M_2} \right) \log P + n \frac{N_1-M_2}{N_1+N_2-M_2}  \alpha_2 \log P  + n \cdot O(1)
\end{align}
which leads to
\begin{align}
  d_1 + \frac{N_1+2N_2-M_2}{N_2} d_2  \le N_1+N_2 +(N_1-M_2)  \alpha_2. 
\end{align}
By exchanging the roles of Receiver~1 and Receiver~2, the outer bound \eqref{eq:ic-outer-bound-7} can be
obtained straightforwardly when $C_2$ holds. 

\section{Conclusion}
In this work, we focus on the two-user MIMO broadcast and interference channels where the transmitter(s) has/have access to both delayed CSIT and an estimate of current CSIT. Specifically, the DoF region of MIMO networks (BC/IC) in this setting with general antenna configuration and general current CSIT qualities has been fully characterized, thanks to a simple yet unified framework employing interference quantization, block-Markov encoding and backward decoding techniques. Our DoF regions generalize a number of existing results under more specific CSIT settings, such as delayed CSIT \cite{Vaze:2011BC,Vaze:2012IC,Ghasemi:2011}, perfect CSIT \cite{MIMOBC2006,Jafar:2007IC}, partial/hybrid/mixed CSIT \cite{RIA,Tandon:2012ISWCS,Vaze:2012Hybrid}. The results further shed light on the benefits of the temporally correlated channel, when such correlation can be opportunistically taken into account for a system design. Finally, we leave some interesting open questions for future works. These include counting in the imperfectness of the delayed CSIT as well as capturing the time-varying nature of  CSIT qualities in MIMO networks.

\appendix

\subsection{Achievable rate regions for the related MAC channels}
\label{app:MAC}

\newcommand{\Qmc}{\Qm_{\text{c}}}
\newcommand{\Qmkc}{\Qm_{k\text{c}}}
\newcommand{\Qmjc}{\Qm_{k\text{c}}}
\newcommand{\Qmk}{\Qm_{k}}
\newcommand{\Qmj}{\Qm_{j}}

\subsubsection{Broadcast Channels}
Let us focus on Receiver $k$,
$k\ne j\in \left\{ 1,2 \right\}$, without loss of generality.
The channel in \eqref{eq:BC-MAC} is a MAC, rewritten as
\begin{align}
  \underbrace{ \begin{bmatrix} \yv_k[b] - \etav_{kb} \\ \etav_{jb}
  \end{bmatrix} }_{Y_k}  &= \underbrace{\begin{bmatrix} \Hm_k \\ 0
  \end{bmatrix}}_{S_1} X_c +  
  \underbrace{\begin{bmatrix} \Hm_k \\ \Hm_j \end{bmatrix} }_{S_2} X_k +
    Z_k
\end{align}%
where $X_c \defeq \xv_c(l_{b-1})$ and $X_k \defeq \uv_k(w_{kb})$ are
independent, with rate $R_c$ and $R_k$, respectively; $Z_k$ is the AWGN. 
It is well known \cite{Cover_Thomas} that a rate pair $(R_c, R_k)$ is achievable in
the channel if 
\begin{align}
  R_c &\le I(X_c; Y_k \cond X_k, S)  \\
  R_k &\le I(X_k; Y_k \cond X_c, S)  \\
  R_c + R_k &\le I(X_c, X_k; Y_k \cond S)  
\end{align}%
for any input distribution $p_{X_c X_k} = p_{X_c} p_{X_k}$; $S\defeq
\left\{ S_1, S_2 \right\}$ denotes the state of the channel. 
Let $X_c\sim\CN[0, \Qmc]$ and
$X_k \sim\CN[0, \Qmk]$ with 
$\Qmc \defeq P \Id_{M}$ and $\Qmk \defeq P^{A'_{k}}
\Phim_{\hat{H}_j^\perp} + P^{A_{k}} \Phim_{\hat{H}_j}$. 
It readily follows that\footnote{Hereafter, we omit for notational brevity the expectation on the channel
states $S$, whenever possible, which does not change the high SNR
behavior in this case. We consider any realization $\Sm_1$ and
$\Sm_2$ instead. }
\begin{align}
I(X_c; Y_k \cond X_k) &= \log \det (\Id + P \Sm_1 \Sm_1^\H) = N_k \log P + O(1)\\
I(X_k; Y_k \cond X_c) &= \log \det (\Id + \Sm_2 \Qmk \Sm_2^\H) =
((M-N_j) A'_k  + N_j A_k) \log P + O(1)
\end{align} 
since $\Sm_2 \in \CC^{(N_1+N_2) \times M}$ has rank $M$ almost
surely, given the assumption $N_1+N_2 \ge M$. For the sum rate constraint, we have
\begin{align}
  I(X_c, X_k; Y_k) &= h(Y_k) - h(Z_k) \label{eq:tmp122}\\ 
&= h(\Hm_j X_k + Z_{k2}) + h(\Hm_k (X_c+X_k) + Z_{k1} \cond
\Hm_j X_k + Z_{k2}) + O(1) \\
&\ge h(\Hm_j X_k + Z_{k2}) + h(\Hm_k (X_c+X_k) + Z_{k1} \cond \Hm_j X_k + Z_{k2}, X_k)  + O(1) \label{eq:tmp123}\\
&= h(\Hm_j X_k + Z_{k2}) + h(\Hm_k X_c + Z_{k1})  +
O(1) \label{eq:tmp124}\\
& = N_j A_k \log P +  N_k \log P + O(1) \label{eq:tmp126}
\end{align}
where we define $Z_{k1}$ and $Z_{k2}$ the first and second parts of the
noise vector $Z_k$; the second equality is from the chain rule and the
fact that the Gaussian noise $Z_k$ is normalized; \eqref{eq:tmp123} is
due to conditioning reduces differential entropy; \eqref{eq:tmp124} is
from the independence between $X_c$ and $X_k$ and between the noises and
the rest; the first term in \eqref{eq:tmp126} is essentially the
differential entropy of the interference $\etav_{jb}$. By relating the
rate pair $(R_c, R_k)$ to the DoF pair $(d_{\eta}, d_{kb})$,
\eqref{eq:MAC-detak}-\eqref{eq:MAC-dsum} is straightforward.   

\subsubsection{Interference Channels}

In \eqref{eq:IC-MAC}, each receiver sees a MAC with three independent
messages. Let us focus on Receiver $k$,
$k\ne j\in \left\{ 1,2 \right\}$, without loss of generality.
The channel in \eqref{eq:IC-MAC} is rewritten as
\begin{align}
  \underbrace{ \begin{bmatrix} \yv_k[b] - \etav_{kb} \\ \etav_{jb}
  \end{bmatrix} }_{Y_k}  &= \underbrace{\begin{bmatrix} \Hm_{kk} \\ 0
  \end{bmatrix}}_{S_{k1}} X_{kc} +  \underbrace{\begin{bmatrix} \Hm_{kj} \\ 0
  \end{bmatrix}}_{S_{k2}} X_{jc} + \underbrace{\begin{bmatrix} \Hm_{kk} \\ \Hm_{jk} \end{bmatrix} }_{S_{k3}} X_{k} + Z_k
\end{align}%
where $X_{kc} \defeq \xv_{kc}(l_{k,b-1})$, $X_{jc} \defeq
\xv_{jc}(l_{j,b-1})$, and $X_k \defeq
\uv_k(w_{kb})$, $k\ne j \in\{1,2\}$, are
three independent signals, with rate $R_{kc}$, $R_{jc}$, and $R_k$, 
respectively; $Z_k$ is the AWGN. 
It is well known \cite{Cover_Thomas} that a rate triple $(R_{kc}, R_{jc}, R_k)$ is achievable in the channel if 
\begin{align}
  R_{kc} &\le I(X_{kc}; Y_k \cond X_{jc}, X_k)  \\
  R_{jc} &\le I(X_{jc}; Y_k \cond X_{kc}, X_k)  \\
  R_k &\le I(X_k; Y_k \cond X_{kc}, X_{jc})  \\
  R_{kc} + R_{jc} &\le I(X_{kc}, X_{jc}; Y_k \cond X_k)  \\
  R_{kc} + R_{k} &\le I(X_{kc}, X_{k}; Y_k \cond X_{jc})  \\
  R_{jc} + R_{k} &\le I(X_{jc}, X_{k}; Y_k \cond X_{kc})
  \label{eq:tmp128} \\
  R_{kc} + R_{jc} + R_{k} &\le I(X_{kc}, X_{jc}, X_{k}; Y_k)
  \label{eq:tmp129}
\end{align}%
for any $p_{X_{kc} X_{jc} X_k} = p_{X_{kc}} p_{X_{jc}} p_{X_k}$; where
we omit the conditioning on the channel states $S$ as in the BC case for
brevity. Let $X_{kc}\sim\CN[0, \Qmkc]$ and $X_k \sim\CN[0, \Qmk]$ with 
$\Qmkc \defeq P \Id_{M_k}$ and $\Qmk \defeq P^{A_{k}}
\Phim_{\hat{H}_{jk}} + P^{A'_{k}}\Phim_{\hat{H}_{jk}^{\perp1}} 
+ P^{A''_{k}} \Phim_{\hat{H}_{jk}^{\perp2}}$. 
It is readily shown that
\begin{align}
  I(X_{kc}; Y_k \cond X_{jc}, X_k) &= \log \det (\Id + P \Sm_{k1} \Sm_{k1}^\H) = \min\left\{ M_k, N_k \right\} \log P + O(1) \\
  I(X_{jc}; Y_k \cond X_{kc}, X_k) &= \log \det (\Id + P \Sm_{k2}
  \Sm_{k2}^\H) = \min\left\{ M_j, N_k \right\} \log P + O(1)
  \label{eq:tmp891}\\
  I(X_{k}; Y_k \cond X_{kc}, X_{jc}) &= \log \det (\Id +  \Sm_{k3}
  \Qmk \Sm_{k3}^\H) = (N'_j A_k + (M_k - N'_j - \xi_k) A'_k + \xi_k
  A''_k) \log P + O(1) \\
  I(X_{jc}, X_{kc}; Y_k \cond X_{k}) &= \log \det (\Id + P \Sm_{k1} \Sm_{k1}^\H + P \Sm_{k2} \Sm_{k2}^\H) =  N_k \log P + O(1)
\end{align} 
since $\Sm_{k3} \in \CC^{(N_1+N_2) \times M_k}$ has rank $M_k$ almost
surely, given the assumption $N_1+N_2 \ge M_k$. Following the same steps
as \eqref{eq:tmp122}-\eqref{eq:tmp126}, we can obtain
\begin{align}
  I(X_{kc}, X_{k}; Y_k \cond X_{jc}) &\ge  N'_j A_k \log P +
  \min\{M_k, N_k\} \log P + O(1). \label{eq:tmp127}
\end{align}%
\begin{figure}[t]
\centering
\includegraphics[width=0.8\columnwidth]{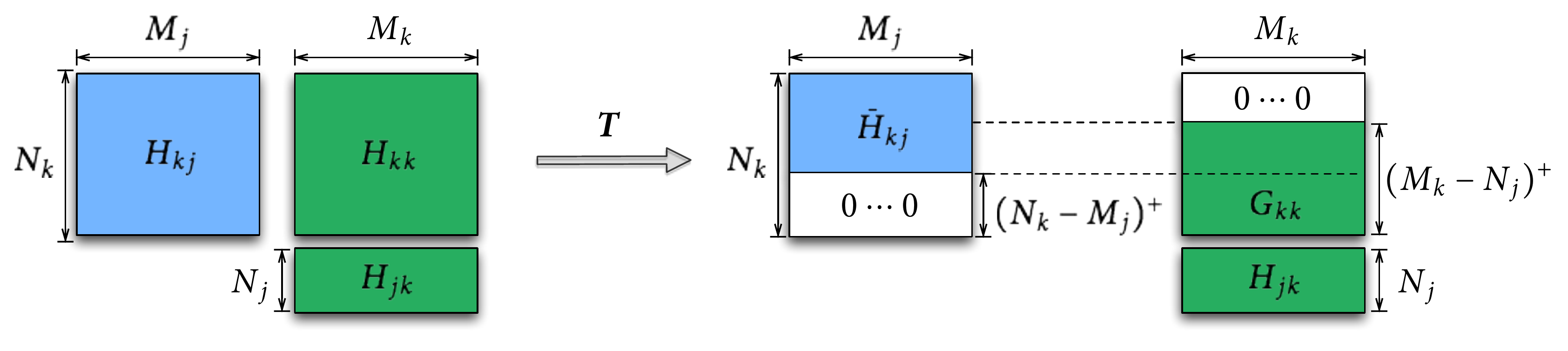}
\caption{Visualization of the interplay between $X_{jc}$ and $X_k$.}
\label{fig:Xk-Xjc}
\end{figure}
It remains to bound the RHS of \eqref{eq:tmp128} and \eqref{eq:tmp129}. 
First, using the chain rule,
\begin{align}
I(X_{jc}, X_{k}; Y_k \cond X_{kc}) &= 
I(X_{k}; Y_k \cond X_{kc}) + I(X_{jc}; Y_k \cond X_{k}, X_{kc})
\label{eq:tmp945}
\end{align}
where the scaling of the second term is already shown in
\eqref{eq:tmp891}. The first term can be interpreted as the rate of
$X_k$ by treating $X_{jc}$ as noise in a two-user MAC with a channel matrix in the block upper triangular form 
$\left[\begin{smallmatrix} \Hm_{kj} & \Hm_{kk} \\ & \Hm_{jk} 
\end{smallmatrix}\right]$. As shown in
Fig.~\ref{fig:Xk-Xjc}, since $\Hm_{kj}$, $\Hm_{kk}$, and
$\Hm_{jk}$ are mutually independent, there exists an invertible
row transformation $\Tm$ that can convert the $(N_1+N_2)\times (M_1+M_2)$
matrix to the form on the right, almost surely. The interference created
by $X_{jc}$ is through the matrix $\bar{\Hm}_{kj}$, only affecting the
overlapping part between $X_{jc}$ and $X_k$, as shown in
Fig.~\ref{fig:Xk-Xjc}. Note that the dimension of the overlapping is
$((M_k - N_j)^+ - (N_k - M_j)^+)^+$ that coincides with the definition of
$\xi_k$ in \eqref{eq:xi}.
From Fig.~\ref{fig:Xk-Xjc}, the interference-free received signal for
$X_k$ is $\tilde{Y}_k = \left[\begin{smallmatrix} \Gm_{kk} \\ \Hm_{jk} 
\end{smallmatrix}\right] X_k + \tilde{Z}_k$. Thus,   
\begin{align}
I(X_{k}; Y_k \cond X_{kc}) &\ge I(X_{k}; \tilde{Y}_k) \\ 
&= \log \det \left(\Id + \left[\begin{smallmatrix} \Gm_{kk} \\
  \Hm_{jk} \end{smallmatrix}\right] \Qmk \left[\begin{smallmatrix} \Gm_{kk} \\
  \Hm_{jk} \end{smallmatrix}\right]^\H \right) + O(1) \label{eq:tmp564} \\
  &\ge \log \det \left(\Id + \left[\begin{smallmatrix} \tilde{\Gm}'_{kk}
    & \tilde{\Gm}_{kk} \\ \tilde{\Hm}'_{kk} & \tilde{\Hm}_{kk}
  \end{smallmatrix}\right] 
  \left[ \begin{smallmatrix}
    P^{A'_k} \Id_{M_k - N'_j - \xi_k} & \\  & P^{A_k} \Id_{N'_j}  
  \end{smallmatrix} \right]
  \left[\begin{smallmatrix} \tilde{\Gm}'_{kk}
    & \tilde{\Gm}_{kk} \\ \tilde{\Hm}'_{kk} & \tilde{\Hm}_{kk}
  \end{smallmatrix}\right]^\H \right) \label{eq:tmp565} + O(1) \\
  &=  (({M_k - N'_j - \xi_k}) A'_k + N'_j A_k) \log P + O(1)
  \label{eq:tmp567}
\end{align}
where the $O(1)$ term in \eqref{eq:tmp564} is from the fact that the covariance of the noise
$\tilde{Z}_k$ depends on $\Tm$ that does not scale with $P$; $\Gm_{kk}$
and $\Hm_{jk}$ remain independent. Next, let $\Qmk = \Um_{jk} \,
\diag(P^{A''_k} \Id_{\xi_k}, P^{A'_k} \Id_{M_k - N'_j - \xi_k}, P^{A_k}
\Id_{N'_j}) \, \Um_{jk}^{\H}$ be the eigenvalue decomposition of $\Qmk$ and
define the column partitions $\bigl[\tilde{\Gm}''_{kk}\,
\tilde{\Gm}'_{kk}\, \tilde{\Gm}_{kk}\bigr] \defeq \Gm_{kk} \Um_{jk}$ and
$\bigl[\tilde{\Hm}''_{jk}\, \tilde{\Hm}'_{jk}\, \tilde{\Hm}_{jk}\bigr]
\defeq \Hm_{jk} \Um_{jk}$ where the number of columns of the
sub-matrices is $\xi_k$, $M_k - N'_j - \xi_k$, and $N'_j$, respectively;
inequality \eqref{eq:tmp565} is from the fact that removing one column block and the
corresponding diagonal block of size $\xi_k$ can only reduce the
log-determinant; the last equality is from the fact that the square
matrix $\left[\begin{smallmatrix} \tilde{\Gm}'_{kk}
  & \tilde{\Gm}_{kk} \\ \tilde{\Hm}'_{jk} & \tilde{\Hm}_{jk}
\end{smallmatrix}\right]$ has full rank, almost surely, for the
following reasons: 1) the matrices $G$ and $H$ are mutually
independent since the column transform $\Um_{jk}$ is unitary and
independent of the $G$ matrices; 2) the rows related to the matrices
$H$ are linearly independent, since it can be proved that
$\rank(\tilde{\Hm}_{jk}) = \rank(\Hm_{jk} \Phim_{\hat{H}_{jk}}
\Hm_{jk}^\H) = \min\left\{ M_k, N_j \right\}$, i.e.,
$\tilde{\Hm}_{jk}$ has full rank; 3) the rows related to the matrices
$G$ are linearly independent as well. Plugging \eqref{eq:tmp567} and
\eqref{eq:tmp891} into \eqref{eq:tmp945}, we have 
\begin{align}
  I(X_{jc}, X_{k}; Y_k \cond X_{kc}) &\ge 
  (N'_k + ({M_k - N'_j - \xi_k}) A'_k + N'_j A_k) \log P + O(1).
\end{align}%
Finally, for the sum rate constraint \eqref{eq:tmp129}, we follow the
same steps as \eqref{eq:tmp122}-\eqref{eq:tmp126}, we can obtain
\begin{align}
  I(X_{kc}, X_{jc}, X_{k}; Y_k) &\ge  N'_j A_k \log P +
  \min\{M_k+M_j, N_k\} \log P + O(1) \\
  &= (N_k + N_j' A_k) \log P + O(1)
\end{align}%

By relating the rate pair $(R_{kc}, R_{jc}, R_k)$ to the DoF pair
$(d_{\eta_1}, d_{\eta_2}, d_{kb})$, \eqref{eq:MAC-dkb-ic}-\eqref{eq:MAC-allsum} are
straightforward.

\subsection{Proof of Lemma~\ref{lemma:gcase}}

In order to prove Lemma~\ref{lemma:gcase}, we provide the following
preliminary results stated as Lemma~\ref{lemma:rank_subm}-\ref{lemma:ELogDet}. 

Let $\Am\in\CC^{N\times M}$, $N\le M$, be a full rank matrix and ${\Am'}\in{\CC^{N\times M'}}$, $M'\le M$, be a submatrix of $\Am$. We have the following lemmas that will be repeatedly used in the rest of the proof.

\begin{lemma}[rank of submatrix]\label{lemma:rank_subm}
  \begin{align}
    \rank(\Am') &\ge \rank(\Am) - (M-M').
  \end{align}%
\end{lemma}

\begin{lemma}\label{lemma:ess_lb}
  Let $\Ic_1,\ldots,\Ic_M$ be a cyclic sliding window of size $N$ on the set of indices $\{1,\ldots,M\}$, i.e.,
  \begin{align}
    \Ic_k \defeq \{(k+i)_M + 1\,: i \in [0,N-1]\},\quad
    k=1,\ldots,M.
  \end{align}%
  If the columns of $\Am$ are arranged such that $\rank(\Am_{\Ic_k}) = N$ for some $k\in[1,M]$, then
  \begin{align}
    \sum_{k=1}^M \mathrm{rank}(\Am_{\Ic_k}) &\ge N^2
  \end{align}%
  where $\Am_{\Ic_k}$ is the matrix composed of $N$ columns of $\Am$ defined by $\Ic_k$, i.e., $\Am_{\Ic_k}\defeq [A_{j,i}]_{j\in[1,N], i\in\Ic_k}$.
\end{lemma}
\begin{proof}
The sketch of the proofs for the above lemma is illustrated in Fig.~\ref{fig:win-a}. Given that there exists $k$ such that the submatrix selected by the window is full rank $N$~(the blue window in Fig.\ref{fig:win-a}), the rank of the submatrix selected by the window $\Ic_{k+1}$ or $\Ic_{k-1}$~(the red window in Fig.\ref{fig:win-b}) is lower bounded by $N-1$. By applying the same argument, it is readily shown that the rank of the submatrix selected by the window $\Ic_{k+2}$ or $\Ic_{k-2}$ is lower bounded by $N-2$. This lower bound keeps decreasing when the window slides away from the blue one, until it hits another lower bound $N-(M-N) = 2N-M$ given by Lemma~\ref{lemma:rank_subm}. The submatrices within the sliding windows are of rank $2N-M$, which lasts $M-1-2(M-N)=2N-M-1$ times. With the help of Fig.\ref{fig:win-a}, a lower bound on the sum of the ranks of all the submatrices visited by the sliding window, can be obtained by counting the dots in the figure, i.e.,
\begin{align}
  N + 2 \sum_{i=1}^{M-N} (N-i) + (2N-M)(2N-M+1) = N^2.
\end{align}
In fact, this can be found easily by ``completing the triangle'', the number of dots in which is $N^2$.
\end{proof}

\begin{lemma}\label{lemma:ess_lb2}
  ${\Am'}\in{\CC^{N\times M'}}$, $N\le M'\le M$, is a submatrix of $\Am$.  We define $\Ic'_1,\ldots,\Ic'_{{M'}}$ as a cyclic sliding window of size $N$ on the set of indices $\{1,\ldots,M'\}$, i.e.,
  \begin{align}
    \Ic'_k \defeq \{(k+i)_{M'}+1\,: i \in
    [0,N-1]\},\quad k=1,\ldots,{M'}.
  \end{align}%
   If the columns of $\Am'$ are arranged such that the first $\rank(\Am')$ columns of $\Am'_{\Ic'_k}$ are linear independent for some $k\in[1,M]$, then we have
  \begin{align}
    \sum_{k=1}^{M'} \mathrm{rank}(\Am'_{\Ic'_k}) &\ge N(N-(M-M'))
  \end{align}%
  where $\Am'_{\Ic'_k}$ is the submatrix of $\Am'$ with $N$ columns defined by $\Ic'_k$, i.e., $\Am'_{\Ic'_k}\defeq [A'_{j,i}]_{j\in[1,N], i\in\Ic'_k}$.
\end{lemma}
\begin{proof}
The sketch of the proofs for the above lemma is illustrated in Fig.\ref{fig:win-b}. Given that there exists $k$ such that the submatrix selected by the window has rank $r = N-(M-M')$ given by Lemma~\ref{lemma:rank_subm} and that the first $r$ columns are linearly independent~(the blue window in Fig.\ref{fig:win-b}), the rank of the submatrix selected by the windows $\Ic'_{k-1}, \ldots, \Ic'_{k-(N-r)}$~(the red and brown windows in Fig.\ref{fig:win-b}) is lower bounded by $r-1$. This lower bound keeps decreasing when the window slides go away from these positions, until it hits another lower bound $N-(M-N) = 2N-M$ given by Lemma~\ref{lemma:rank_subm}. With the help of Fig.\ref{fig:win-b}, a lower bound on the sum of the ranks of all the submatrices visited by the sliding window, can be obtained by counting the dots in the Figure. In fact, after some basic computations, it turns out that there are $N(N-(M-M'))$ dots.
\end{proof}

\begin{figure*}
\centering
  \begin{subfigure}{0.47\textwidth}
  \includegraphics*[width=\columnwidth]{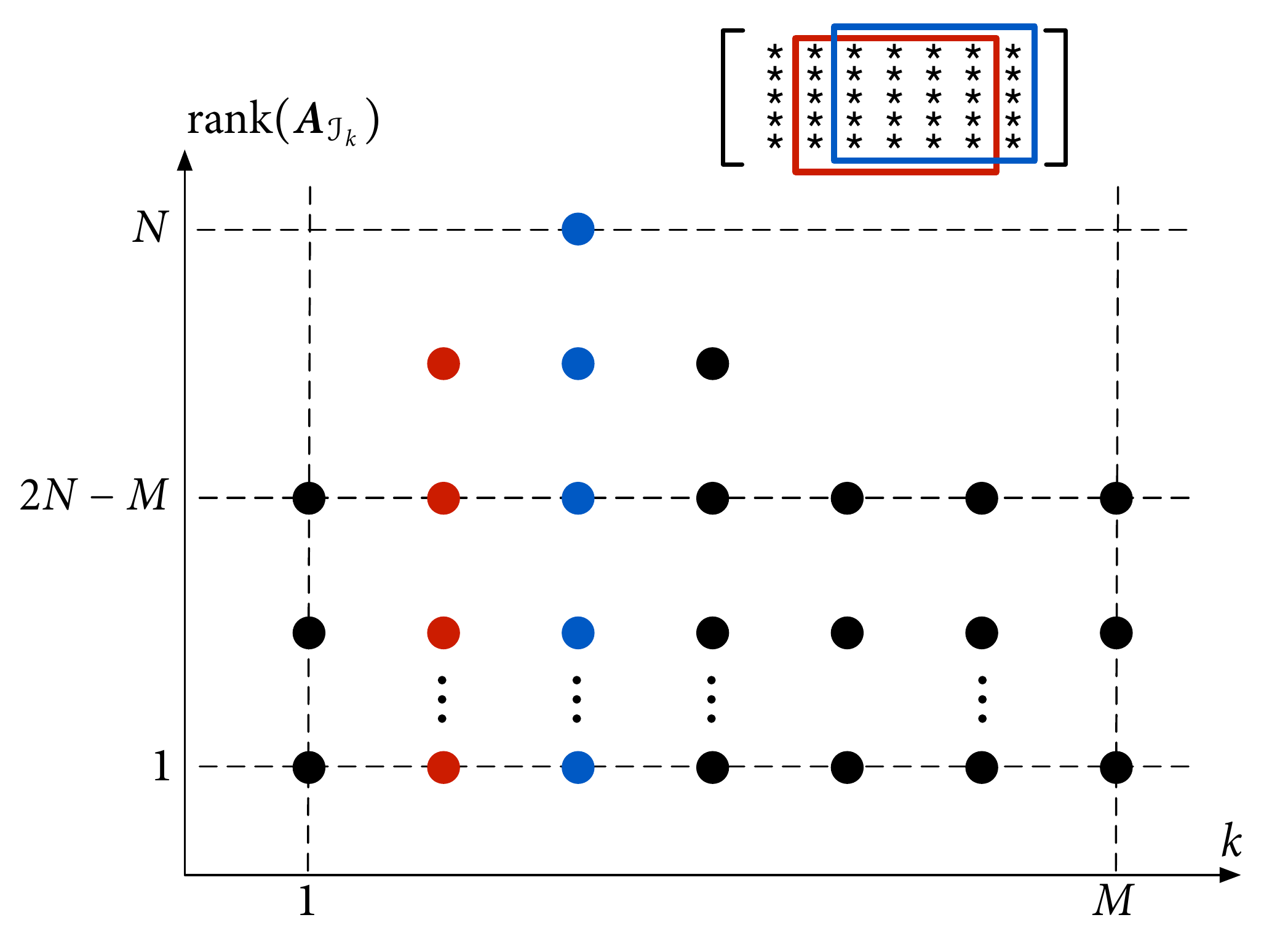}  
  \caption{An example for Lemma~\ref{lemma:ess_lb}}
  \label{fig:win-a}
  \end{subfigure}  
  \begin{subfigure}{0.50\textwidth}
  \includegraphics*[width=\columnwidth]{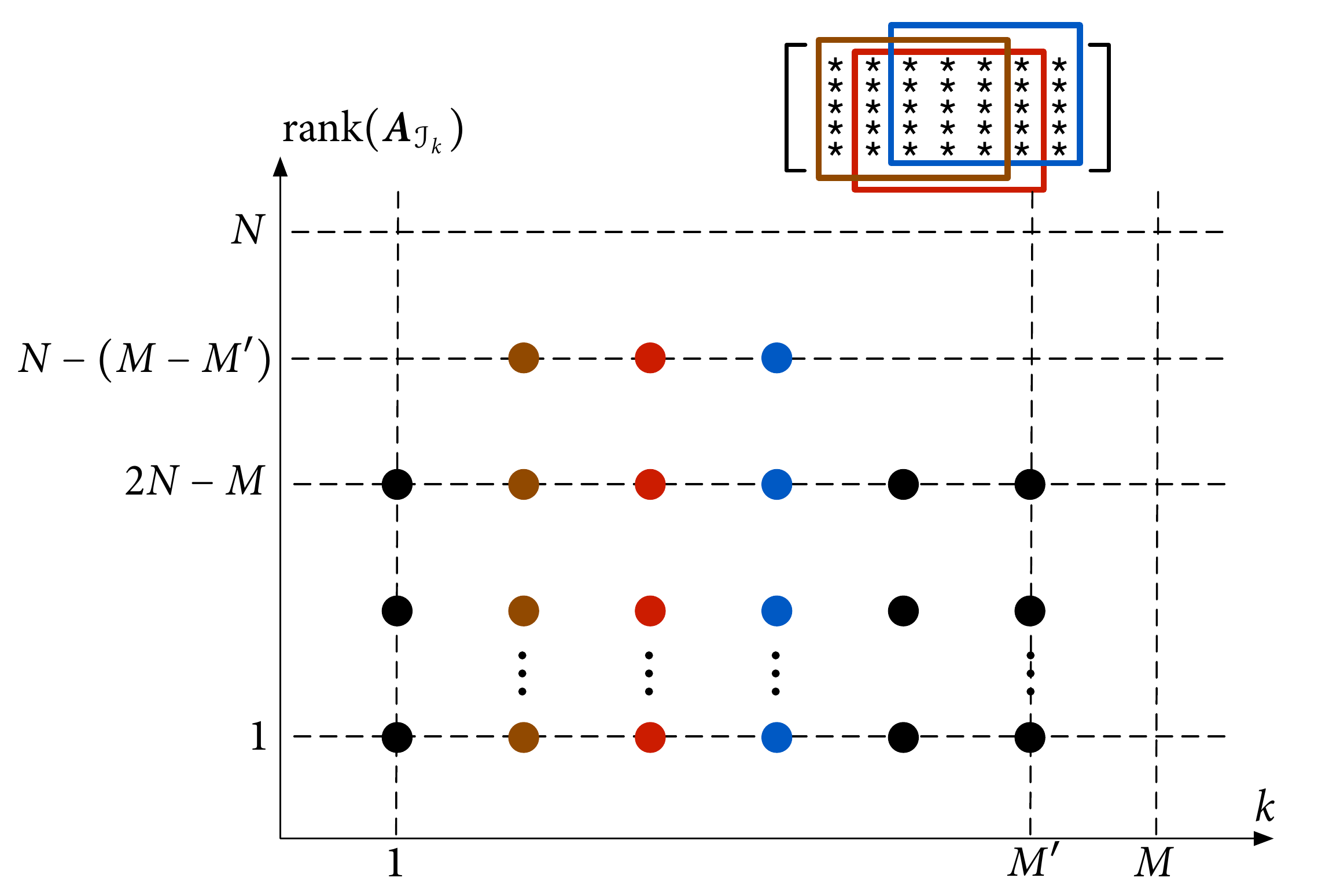} 
  \caption{An example for Lemma~\ref{lemma:ess_lb2}}
  \label{fig:win-b}
  \end{subfigure}
  \caption{Illustrations of the worst-case ranks of the submatrices from a sliding
  window. For each $k$, the number of vertical dots represents the rank
  of the submatrix $\Am_{\Ic_k}$. In particular, the number of
  red~(resp. blue) dots
  is the rank of the submatrix selected by the red~(resp. blue) window.
  The sum of the ranks can be found by counting the number of dots.}
\end{figure*}

\begin{lemma} \label{lemma:ELogDet}
  Given $\Hm = \hat{\Hm} + \tilde{\Hm} \in \CC^{N\times M}$, $N\le M$,
  with the entries of $\tilde{\Hm}$ being i.i.d.~$\CN[0,\sigma^2]$,
  $\sigma>0$, 
  then
  \begin{align}
    \E_{\tilde{\Hm}}\log \det (\Hm\Hm^\H) &\ge (N-\rank(\hat{\Hm})) \log
    \sigma^2 + \Eh
  \end{align}%
  as $\sigma^2$ goes to $0$.
\end{lemma}
\begin{proof}
  According to \cite[Lemma~1]{Chen:2013}, for any $\Gm =
  \hat{\Gm} + \tilde{\Gm} \in \CC^{N\times N}$ with the entries in
  $\tilde{\Gm}$ i.i.d.~$\CN[0,1]$ independent of $\hat{\Gm}$, we have
  \begin{align}
    \E_{\tilde{G}} \log \det (\Gm \Gm^\H) &\ge
    \sum_{i=1}^{ \tau }\log(\lambda_i(\hat{\Gm}\hat{\Gm}^\H)) + O(1)
  \end{align}%
  where $\tau\le\rank(\hat{\Gm})$ is the number of eigenvalues of $\Gm$ that are 
  larger than $1$. From here, it follows that
  \begin{align}
    \E_{\tilde{G}} \log \det (\Gm \Gm^\H) &\ge
    \sum_{i=1}^{ \rank(\hat{\Gm}) }\log(1+\lambda_i(\hat{\Gm}\hat{\Gm}^\H)) + O(1)
  \end{align}%
  since the remaining $\rank(\hat{\Gm}) - \tau$ eigenvalues are smaller
  than $1$ and do not contribute more than $O(1)$ to the expectation. 
  Therefore, for any $\sigma>0$, we can
  apply the above inequality to $\sigma^{-1} \Hm$ and have
  \begin{align}
    \E_{\tilde{\Hm}}\log \det ((\sigma^{-1}\Hm)(\sigma^{-1}\Hm)^\H) &\ge
    \sum_{i=1}^{ \rank(\hat{\Hm}) }\log(\lambda_i(\sigma^{-2}\hat{\Hm}\hat{\Hm}^\H)) + O(1) \\
    &= -\rank(\hat{\Hm}) \log \sigma^2 + \sum_{i=1}^{ \rank(\hat{\Hm}) }\log(\lambda_i(\hat{\Hm}\hat{\Hm}^\H)) + O(1) \\
    &= -\rank(\hat{\Hm}) \log \sigma^2 + \Eh  
  \end{align}%
  where the last equality is from the Assumption~\ref{ass:CSI} that 
  $\E_{\hat{H}}(\log\det(\hat{\Hm}\hat{\Hm}^\H)) > -\infty$. 
\end{proof}

In the following, we prove Lemma~\ref{lemma:gcase} case by case
according to the value of $M$\footnote{The technique employed in this
proof was first developed in our earlier version of this paper
\cite{Yi:2012}, and later applied and extended to tackle the $K$-user
 MISO case in
\cite{Kerret:2013,Chen:2013}.}.  First, let us recall that $N\le L$.
Since the case with $M\le N$ is trivial, we focus on the cases with $N <
M < L$ and $M \ge L$.

\subsubsection{{Case A: $N < M < L$}}

Let us define $M'$ as the number of eigenvalues of $\Km$ that are not
smaller than $1$\footnote{Or any constant $c>0$ that is independent of
any parameter in the system. Note that $M'$ can depend on
$\hat{\Sm}$ and the SNR $P$.}, and let $\Km = \Vm \Lambdam \Vm^\H$ be the eigenvalue decomposition of $\Km$. We first establish the following upper bound:
\begin{align}
  \det(\Id + \Sm\! \Km\! \Sm^\H) &= \det(\Id + \Lambdam \Vm^\H \Sm^\H\Sm
  \Vm^\H ) \\
  &\le \det(\Id + \lambda_{\max}(\Vm^\H\Sm^\H\Sm \Vm)\Lambdam ) \\
  &\le \det(\Id + \|\Sm\|^2_{\text{F}} \, \Lambdam)  \label{eq:tmp57}
\end{align}%
where the last inequality is due to $\lambda_{\max}(\Vm^\H\Sm^\H \Sm \Vm) \le
\Normf{\Sm\Vm} = \Normf{\Sm}$. 
Therefore, we have
\begin{align}
  \E_{\tilde{\Sm}} \log \det(\Id + \Sm\! \Km\! \Sm^\H) 
  &\le \log \det(\Id + \E_{\tilde{\Sm}} (\Normf{\Sm}) \, \Lambdam)
  \\
  &\le \log \det(\Id + \E_{\tilde{\Sm}} (\Normf{\Sm})\, \Lambdam') + \Es \label{eq:ub} 
\end{align}%
where the first inequality is from \eqref{eq:tmp57} on which we apply
Jensen's inequality; 
$\Lambdam'$ is a diagonal matrix composed of the $M'$ largest eigenvalues of $\Km$.  

Next, let $\Phim \defeq \Hm \Vm$, $\Phim' \defeq \hat{\Hm} \Vm$. Without loss of generality, we assume that the columns of $\Phim$ and $\Phim'$ are arranged such that the conditions in Lemma~\ref{lemma:ess_lb} and Lemma~\ref{lemma:ess_lb2} are satisfied (i.e., $\rank(\Phim_{\Ic})=N$, where $\Ic$ is the cyclic window with size $N$, and $\Phim_{\Ic}$ is defined as in Lemma~\ref{lemma:ess_lb}), respectively. This also implies that the eigenvalues in $\Lambdam$ and $\Lambdam'$ are arranged accordingly. In the following, given different values of $M'$, we prove that
\begin{align} \label{eq:lb}
\E_{\tilde{\Hm}} \log \det(\Id + \Hm\! \Km\! \Hm^\H) \ge \frac{N}{M}
\log \det(\Lambdam') + \frac{N(M-N)}{M} \log \sigma^2 + \Eh.
\end{align}

\subsubsection*{\underline{Case $M' = M$}}
In this case, we have
\begin{align}
  \det(\Id + \Hm\! \Km\! \Hm^\H) &=
  \det(\Id + \Phim \Lambdam \Phim^\H) \\
  &= \sum_{\Ic\subseteq \{1,\ldots,N\}} \det(\Lambdam_{\Ic})
  \det(\Phim_{\Ic}^\H \Phim_{\Ic}) \label{eq:tmp96}\\
  &\ge \sum_{k=1}^M \det(\Phim_{\Ic_k}^\H\Phim_{\Ic_k})
  \det(\Lambdam_{\Ic_k}) \label{eq:tmp97}
\end{align}%
where \eqref{eq:tmp96} is an application of the
identity $\det(I + \Am) = \sum_{\mathcal{I}\subseteq\left\{ 1,\ldots,M
\right\}} \!\!  \det(\Am_{\mathcal{I}\mathcal{I}})$ for any $\Am\in
\CC^{M\times M}$~\cite{HJ:90}; 
 the lower bound is obtained by only considering a sliding window of
 size $N$ for all the possible sub-determinant. Thus,
\begin{align}
  \log \det(\Id + \Hm\! \Km\! \Hm^\H) &\ge \log \left(\sum_{k=1}^M
  \det(\Phim_{\Ic_k}^\H\Phim_{\Ic_k})
  \det(\Lambdam_{\Ic_k})\right) \label{eq:tmp98}\\
  &\ge \log \left(\frac{1}{M} \sum_{k=1}^M
  \det(\Phim_{\Ic_k}^\H\Phim_{\Ic_k})
  \det(\Lambdam_{\Ic_k})\right) \\
  &\ge  \frac{1}{M} \log \left(\prod_{k=1}^M
  \det(\Phim_{\Ic_k}^\H\Phim_{\Ic_k})
  \det(\Lambdam_{\Ic_k})\right) \label{eq:tmp99}\\
  &= \frac{1}{M} \left( N \log \det(\Lambdam) + \sum_{k=1}^M \log \det(\Phim_{\Ic_k}^\H\Phim_{\Ic_k}) \right)
\end{align}%
where \eqref{eq:tmp99} holds since arithmetic mean is not smaller than
geometric mean; the last equality is from the sliding window property $\prod_{k=1}^M \det(\Lambdam_{\Ic_k}) = \det(\Lambdam)^N$. Finally, we have
\begin{align}
  \E_{\tilde{\Hm}} \log \det(\Id + \Hm\! \Km\! \Hm^\H)
  &\ge \frac{1}{M} \left( N \log \det(\Lambdam) + \sum_{k=1}^M
  \E_{\tilde{\Hm}} \log \det(\Phim_{\Ic_k}^\H\Phim_{\Ic_k}) \right)
  \label{eq:tmp887} \\
  &\ge \frac{1}{M} \left( N \log \det(\Lambdam) + \log \sigma^2
  \sum_{k=1}^M (N-\rank(\hat{\Phim}_{\Ic_{k}}))  + \Eh \right)
  \label{eq:tmp87} \\
  &= \frac{1}{M} \left( N \log \det(\Lambdam) + \log \sigma^2  \left(MN
  - \sum_{k=1}^M \rank(\hat{\Phim}_{\Ic_{k}})\right)  \right) + \Eh \label{eq:tmp86} \\
  &\ge \frac{N}{M} \left( \log \det(\Lambdam') + (M-N) \log \sigma^2
  \right) + \Eh \label{eq:tmp81}
\end{align}%
where $\hat{\Phim} \defeq \hat{\Hm}\Vm$ and hence $\rank(\hat{\Phim})
=\rank(\hat{\Hm})$; \eqref{eq:tmp87} is from Lemma~\ref{lemma:ELogDet};
the last inequality is from Lemma~\ref{lemma:ess_lb} and that $\Lambdam=\Lambdam'$ as $M=M'$. 

\subsubsection*{\underline{Case $M > M' \ge N$}}
For this case, we can first get
\begin{align}
  \det(\Id + \Hm\! \Km\! \Hm^\H) &= \det(\Id + \Phim \Lambdam \Phim^\H) 
  \ge \det(\Id + \Phim' \Lambdam' (\Phim')^\H). \label{eq:tmp66}
\end{align}%
Following the same footsteps as in \eqref{eq:tmp887}-\eqref{eq:tmp86}, we obtain
\begin{align}
  \E_{\tilde{\Hm}} \log \det(\Id + \Hm\! \Km\! \Hm^\H)
  &\ge \frac{1}{M'} \left( N \log \det(\Lambdam') + \log \sigma^2
  \left(M'N - \sum_{k=1}^{M'} \rank(\hat{\Phim}'_{\Ic'_{k}})\right)
  \right) + \Eh \label{eq:tmp76} \\
  &\ge \frac{N}{M'} \left( \log \det(\Lambdam') + (M-N) \log \sigma^2
  \right)  + \Eh \label{eq:tmp71}\\
  &\ge \frac{N}{M}  \log \det(\Lambdam') + \frac{N(M-N)}{M} \log
  \sigma^2 + \Eh
\end{align}
where the inequality \eqref{eq:tmp71} is from Lemma~\ref{lemma:ess_lb2}. 

\subsubsection*{\underline{Case $M' < N$}}
From \eqref{eq:tmp66} and given that $M'<N$, we have
\begin{align}
  \E_{\tilde{\Hm}} \log \det(\Id + \Hm\! \Km\! \Hm^\H)
  &\ge \log \det(\Lambdam') + \log \sigma^2 \left(M' -
  \rank(\hat{\Phim}')\right) + \Eh \label{eq:tmp67} \\
  &\ge \log \det(\Lambdam') + \log \sigma^2 \left(M' - (N-(M-M'))
  \right) + \Eh \label{eq:tmp69} \\
  &= \log \det(\Lambdam') + (M-N) \log \sigma^2 + \Eh \label{eq:tmp61}\\
  &\ge \frac{N}{M}  \log \det(\Lambdam') + \frac{N(M-N)}{M} \log
  \sigma^2 + \Eh
\end{align}
where \eqref{eq:tmp69} is from $\log\sigma^2 \le 0$ and $\rank(\hat{\Phim}') \ge N - (M-M')$. 

By now, \eqref{eq:lb} has been proved in all configurations of $(M, N, M')$.  
Combining \eqref{eq:ub} and \eqref{eq:lb}, we have
\begin{align}
  \MoveEqLeft N\, \E_{\tilde{\Sm}} \log \det(\Id + \Sm\! \Km\! \Sm^\H) -
  M\, \E_{\tilde{\Hm}} \log \det(\Id + \Hm\! \Km\! \Hm^\H) \nonumber \\
&\le  -N(M-N)\log\sigma^{2} + N\,\log
{\det\bigl(\E_{\tilde{\Sm}}(\Normf{\Sm})\, \Id +(\Lambdam')^{-1}\bigr)} 
+ \Es + \Eh  \label{eq:tmp909} \\
&\le  -N(M-N)\log\sigma^{2} + \Es + \Eh  \label{eq:tmp910}
\end{align}%
where the last inequality is from the fact that $\Lambdam'\succeq \Id$ by construction
and hence $\log {\det\bigl(\E_{\tilde{\Sm}}(\Normf{\Sm})\, \Id
+(\Lambdam')^{-1}\bigr)} \le M'
\log(1+\E_{\tilde{\Sm}}(\Normf{\Sm}))=\Es$. This completes
the proof of \eqref{eq:lemma3} for the case $N < M < L$.

\subsubsection{{Case B: $M \ge L$}}
For the first term in \eqref{eq:lemma3}, we bound it as follows
  \begin{align}
    \E_{\tilde{\Sm}} \log \det (\Id+\Sm \Km \Sm^\H) &= \E_{\tilde{\Sm}} \log \det (\Id+\Um_{\Sm}  \mathbf{\Sigma}_{\Sm} \Vm_{\Sm}^\H \Km \Vm_{\Sm} \mathbf{\Sigma}_{\Sm} \Um_{\Sm}^\H) \label{eq:Theorem-2-1}  \\
    &=\E_{\tilde{\Sm}} \log \det (\Id+\mathbf{\Sigma}_{\Sm}^2 \Vm_{\Sm}^\H \Km \Vm_{\Sm} ) \label{eq:Theorem-2-2}\\
    &\le \E_{\tilde{\Sm}} \log \det (\Id+\lambda_{\max}(\mathbf{\Sigma}_{\Sm}^2) \Vm_{\Sm}^\H \Km \Vm_{\Sm} ) \\
    &=\sum_{i=1}^L \E_{\tilde{\Sm}} \log  (1+\lambda_{\max}(\Sm \Sm^\H) \lambda_{i}(\Vm_{\Sm}^\H \Km \Vm_{\Sm}))\\
    &\le \sum_{i=1}^L \E_{\tilde{\Sm}} \log  (1+\lambda_{\max}(\Sm \Sm^\H) \lambda_{i}) \label{eq:Theorem-2-3}\\
    &\le \sum_{i=1}^L \E_{\tilde{\Sm}} \log  (1+\Normf{\Sm} \lambda_{i}) \label{eq:Theorem-2-4}\\
&\le \sum_{i=1}^{L} \log  (1 + \E_{\tilde{\Sm}}(\Normf{{\Sm}})\, \lambda_i) \label{eq:Theorem-2-5}\\
    &= \log \det (\Id + \E_{\tilde{\Sm}}(\Normf{{\Sm}})\, \Lambdam'') \\
    &= \log \det (\Id + \E_{\tilde{\Sm}}(\Normf{{\Sm}})\, \Lambdam''') +
    \Es
  \end{align}
  where in (\ref{eq:Theorem-2-1}), $\Sm=\Um_{\Sm}  \mathbf{\Sigma}_{\Sm}
  \Vm_{\Sm}^\H$ with $\mathbf{\Sigma}_{\Sm} \in \CC^{N \times N}$ and
  $\Vm_{\Sm} \in \CC^{M \times L}$; (\ref{eq:Theorem-2-2}) comes from
  the equality $\det(\Id+\Am\Bm)=\det(\Id+\Bm\Am)$;
  (\ref{eq:Theorem-2-3}) is due to Poincare Separation
  Theorem~\cite{HJ:90} that $\lambda_{i}(\Vm_{\Sm}^\H \Km \Vm_{\Sm}) \le
  \lambda_{i} (\Km)$ for $i=1,\cdots,N$; (\ref{eq:Theorem-2-4}) is from
  the fact that $\lambda_{\max}(\Sm\Sm^\H) \le \Normf{\Sm}$;
  (\ref{eq:Theorem-2-5}) is obtained by applying Jensen's inequality;
  $\Lambdam'' \defeq \diag (\lambda_1, \cdots, \lambda_{L})$ and $\Lambdam''' \defeq \diag (\lambda_1, \cdots, \lambda_{\min\{L,M'\}})$
  with $M'$ being the number of eigenvalues that are not smaller than
  $1$, i.e., $\Lambda'''\succeq \Id$.

  For the second term in \eqref{eq:lemma3}, we use the following lower
  bound
   \begin{align}
    \E_{\tilde{\Hm}} \log \det (\Id+\Hm \Km \Hm^\H) &=\E_{\tilde{\Hm}} \log \det (\Id+ \mathbf{\Phi} \mathbf{\Lambda}  \mathbf{\Phi}^\H )\\
    &\ge \E_{\tilde{\Hm}} \log \det (\Id+ \mathbf{\Phi}' \mathbf{\Lambda}''  \mathbf{\Phi}'^\H ) \label{eq:Theorem-2-12}\\
    &\ge \frac{N}{L} \log \det (\Lambdam''') +\frac{N(L-N)}{L}
    \log(\sigma^2) + \Eh
  \end{align}
  where  $\Phim\defeq\Hm \Vm \in \CC^{{N} \times M}$ with $\Vm$ being the unitary
  matrix containing the eigenvectors of $\Km$, i.e., $\Km=\Vm \Lambdam
  \Vm^\H$ with $\Lambdam=\diag(\lambda_1,\cdots,\lambda_M)$; in
  (\ref{eq:Theorem-2-12}), ${\Phim}'=\Hm \Vm' \in \CC^{{N} \times {L}}$
  with $\Vm'$ being the first ${L}$ columns of $\Vm$, and multiplying by
  the matrix $\Vm'$ does not change the distribution property, and
  ${\Phim} {\Lambdam}  {\Phim}^\H \succeq {\Phim}' {\Lambdam}''
  {\Phim}'^\H$; the last inequality is obtained from \eqref{eq:lb} in
  the previous subsection.
  
Finally, it is readily shown that, following the same steps as in
\eqref{eq:tmp909} and \eqref{eq:tmp910}, 
\begin{align}
\frac{1}{L}\E_{\tilde{\Sm}} \log \det (\Id+\Sm \Km \Sm^\H) -
\frac{1}{N}\E_{\tilde{\Hm}} \log \det (\Id+\Hm \Km \Hm^\H) \le -
\frac{L-N}{L} \log(\sigma^2) + \Es + \Eh.
\end{align}

This completes the proof of \eqref{eq:lemma3} for the case $M \ge L$.

\end{document}